\documentclass[journal,draftclsnofoot,onecolumn,12pt,twoside]{IEEEtranTCOM}
\normalsize

\usepackage{cite}
\usepackage{graphicx}
\usepackage{epstopdf}  
\usepackage{amsmath, amsthm, amssymb}
\usepackage{balance}

\newtheorem{lemma}{Lemma}
\newtheorem{proposition}{Proposition}

\usepackage{xcolor} 

\usepackage{subcaption} 
\renewcommand{\baselinestretch}{1.36}
\begin{document}
%
\title{Secure Swarm UAV-assisted Communications with Cooperative Friendly Jamming}
%
%
%

\author{Hanh~Dang-Ngoc,~\IEEEmembership{}
        Diep~N.~Nguyen,~\IEEEmembership{}
        Khuong~Ho-Van,~\IEEEmembership{} \\
        Dinh~Thai~Hoang,~\IEEEmembership{}
        Eryk~Dutkiewicz,~\IEEEmembership{}
        Quoc-Viet~Pham,~\IEEEmembership{}
        and~Won-Joo~Hwang~\IEEEmembership{}
\thanks {Hanh Dang-Ngoc, Diep  N.  Nguyen,  Dinh  Thai  Hoang,  and  Eryk  Dutkiewicz  are  with  the  School  of  Electrical  and  Data  Engi-neering, University of Technology Sydney, NSW 2007, Australia (e-mail: hanh.n.dang@student.uts.edu.au; diep.nguyen@uts.edu.au;  hoang.dinh@uts.edu.au; 
eryk.dutkiewicz@uts.edu.au).}%
\thanks {Hanh Dang-Ngoc and Khuong Ho-Van are with the Ho Chi Minh City University of Technology, VNU, Vietnam (e-mail: hanhdn@hcmut.edu.vn; hvkhuong@hcmut.edu.vn).}%
\thanks {Quoc-Viet Pham is with the Korean Southeast Center for the 4th Industrial Revolution Leader Education, Pusan National University, Busan 46241, Republic of Korea (e-mail: vietpq@pusan.ac.kr).}%
\thanks {Won-Joo Hwang is with the Department of Biomedical Convergence Engineering, Pusan National University, Yangsan 50612, Republic of Korea (e-mail: wjhwang@pusan.ac.kr).}%
}


\maketitle

\begin{abstract}
This article proposes a cooperative friendly jamming framework for swarm unmanned aerial vehicle (UAV)-assisted amplify-and-forward (AF) relaying networks with wireless energy harvesting. In particular, we consider a swarm of hovering UAVs that relays information from a terrestrial source to a distant destination and simultaneously generates ``friendly" jamming signals to interfere an eavesdropper. Due to the limited energy of the UAVs, we develop a collaborative time-switching relaying protocol which allows the UAVs to collaborate to harvest wireless energy, relay information, and jam the eavesdropper. To evaluate the secrecy rate, we derive the secrecy outage probability (SOP) for two popular detection techniques at the eavesdropper, i.e., selection combining and maximum-ratio combining. Monte Carlo simulations are then used to validate the theoretical SOP derivation. Using the derived SOP, one can obtain engineering insights to optimize the energy harvesting time and the number of UAVs in the swarm to achieve a given secrecy protection level. Furthermore, Monte Carlo simulations show the effectiveness of the proposed framework in terms of SOP as compared with the conventional amplify-and-forward relaying system. The analytical SOP derived in this work can also be helpful in future UAV secure-communications optimizations (e.g., trajectory, locations of UAVs). As an example, we present a case study to find the optimal corridor to locate the swarm so as to minimize the system SOP.
\end{abstract}


%
\IEEEpeerreviewmaketitle

\section{Introduction}
Unmanned Aerial Vehicles (UAVs), thanks to their high mobility and flexibility, have great potential applications to future communications systems \cite{y_zeng_cellular_connected_2019,a_fotouhi_survey_2019,2019_mozaffari_a_tutorial}. 
{Carrying a light base station (BS) on-board, {a UAV} can be deployed as a flying BS to {replace a} terrestrial BS in certain challenging scenarios. {Moreover, UAVs} can be utilized to connect ground Internet of things \cite{2021_abouzaid_the_meshing} or to establish communications in case of damaged infrastructure in disaster regions \cite{2018_ur_rahman_positioning,2019_liu_resource}.}
In practical UAV-assisted cellular networks, {aerial BSs} can be deployed as relays to assist the communications between terrestrial nodes, {which are suffering from the absence of direct links} 
\cite{2020_hosseinalipour_interference,2018_chen_using,2021_han_towards}. 
{A single UAV was used to form dual-hop relay mode \cite{2020_hosseinalipour_interference}, while multiple UAVs were used to form either a single multi-hop link or multiple dual-hop links \cite{2018_chen_using}. In these studies, the placement of the UAVs was optimized to maximize the signal-to-noise ratio or the average signal-to-interference ratio of the system.
}	
%

%
Due to the broadcast nature of line-of-sight (LoS) dominated aerial based wireless communications, the UAVs' communication links are vulnerable to eavesdropping {or jamming}. The authors in \cite{x._yuan_secrecy_2020} proved that using relays also offers {eavesdroppers} the higher intercepting level by analyzing a relaying network between the on-ground legitimate source and destination in the presence of cooperative UAV-enabled eavesdroppers which use the maximum ratio combining (MRC) or the selection combining (SC) scheme. 
Friendly jamming has been introduced as a promising technology to alleviate the problem {in UAV-enabled {networks} \cite{y._zhou_improving_2018,w_wang_energy_constrained_2020}. In \cite{y._zhou_improving_2018}, UAVs were deployed to send {jamming signals, resulting in a decrease in} the suspicious rate of an amplify-and-forward (AF) multi-relay system.}
{In addition, having a degree of freedom to change locations over time, the flying paths and locations of UAVs can be optimized to enhance the physical layer security of wireless {networks} \cite{w_wang_energy_constrained_2020,2020_wu_energy}, {i.e.,} retreating away from the jammers or eavesdroppers.}
%

%
{The security in UAV-assisted relaying systems has been recently studied in \cite{r_ma_secure_2019,m_tatar_mamaghani_performance_2019,w_wang_energy_constrained_2020,2020_miao_cooperative}.
To assist and secure the system in the presence of an eavesdropper, a relaying UAV and a jamming UAV were utilized in \cite{2020_miao_cooperative}. The transmit power and flight trajectory of both relaying UAV and jamming UAV were optimized to maximize the secrecy rate of the system. In this study, the authors considered that the relaying UAV forwards the received information without decoding or {amplifying,} and the eavesdropper overhears the information in the relaying phase.} 
%
In \cite{r_ma_secure_2019}, a swarm of UAVs {was} divided into two groups of decode-and-forward relays and jammers.  Without considering the limitation of UAV's on-board energy, the secrecy outage probability (SOP) was analyzed {when the} eavesdroppers only {listen} to the relay communications phase. 
%

Given the limitation of on-board energy, UAVs can be equipped with solar panels to harvest solar energy \cite{y_sun_optimal_2019}, or with energy harvesters which can scavenge energy in the radio frequency (RF) signals to self-power their signal transmission \cite{d_n_k_jayakody_self-energized_2020,2021_chittoor_a_review}.
%
	{The authors in \cite{m_tatar_mamaghani_performance_2019,w_wang_energy_constrained_2020} studied a wireless information and power transfer system which employs energy-constrained aerial node as a relay and the full-duplex destination nodes to transmit artificial noise to confuse the malicious eavesdroppers. However, using full-duplex on-ground destination nodes as in \cite{m_tatar_mamaghani_performance_2019,w_wang_energy_constrained_2020} {might not be effective} in practical cases with severe obstacles, long distance, and deep fading. The self-interference in full-duplex radios can also have an adverse impact on signal reception/decoding at legitimate receivers.} 

%

In this work, we propose a cooperative friendly jamming framework for swarm UAV-assisted AF relaying networks with wireless energy harvesting (EH) ability. 
In particular, we consider a swarm of hovering UAVs that {relays} information from a terrestrial source to a distant destination and simultaneously {generates} ``friendly" jamming signals to interfere {an} eavesdropper. Due to the limited energy of the UAVs, we develop a collaborative time-switching relaying (TSR) protocol which allows the UAVs to collaborate to harvest wireless energy, relay information, and jam the eavesdropper.
	%
{Alternatively,} to conserve on-board energy, we assume that UAVs {operate} in the half-duplex mode (i.e., not equipped with the self-interference suppression capability) in which {they receive the information from the source and jam the eavesdropper in two separate phases}.  
During these phases, the eavesdropper can intercept the information from both the source and the relay UAV using either {SC or MRC combining scheme}. 
To evaluate the secrecy rate, we derive the SOP and then use Monte Carlo simulations to validate the theoretical SOP derivation.  

In practice, an eavesdropper is often a passive device (i.e., not emitting signal), the channel state information (CSI) between it and the legitimate transmitter is often unknown. For that matter, we further extend our study to the case of a randomly distributed eavesdropper with unknown CSI. Using the derived SOP, one can obtain engineering insights to optimize the energy harvesting time, the number of UAVs in the swarm, as well as their placements, to achieve a given secrecy protection level. As a {case study}, we apply the analysis above to find the optimal corridor for locating UAVs to minimize the system SOP in the presence of an eavesdropper. 
The major contributions of our work are as follows. 
\begin{itemize}
\item Propose an effective model and protocol to {utilize} a swarm of wireless-powered UAVs {to simultaneously} relay information and to jam the eavesdropper, under a practical shadowed-Rician fading model.   
\item Derive the expressions of the system SOP for two cases of SC and MRC combining techniques at the eavesdropper.
\item Conduct the Monte Carlo simulations to verify the expressions and illustrate the impact of system parameters such as the EH time factor and the number of UAVs. 
\item Present a practical {case study} of finding the optimal corridor for locating UAVs in the three-dimensional (3D) space to minimize the system SOP in the presence of an eavesdropper. 
\end{itemize}
%

The rest of the paper is as follows. In Section~\ref{Sec:SystemModel}, we describe our proposed system model and derive related signal-to-noise ratios (SNRs). The analytical expressions of the system SOP under SC and MRC at the eavesdropper are presented in Section~\ref{Sec:Secrecy}. The SOP is analyzed for the case of a randomly distributed eavesdropper with unknown CSI in Section~\ref{Sec:Unknown}. In Section~\ref{Sec:Case}, as an example of the application of the derived SOP, we provide a {case study} in which we use SOP to find the optimal space for UAVs so as to guarantee a given level of SOP. Numerical results are presented and discussed in Section~\ref{Sec:Results}. Finally, conclusions are drawn in Section~\ref{Sec:Conclusion}.
%

\textit{Notations}: 
${|\cdot|}$ is the Euclidean norm; $f_X \left(\cdot \right)$ and $F_X \left(\cdot \right)$ denote the probability density functions (PDF) and the cumulative distribution function (CDF) of the random variable (r.v.) $X$, respectively; $\mathbb{E}\left( \cdot  \right)$ is the statistical expectation operation; the operation $\text{Pr}\left( \cdot  \right)$ returns probability.
\begin{figure}[t!]
    \begin{center}
	\includegraphics[width=0.85\linewidth]{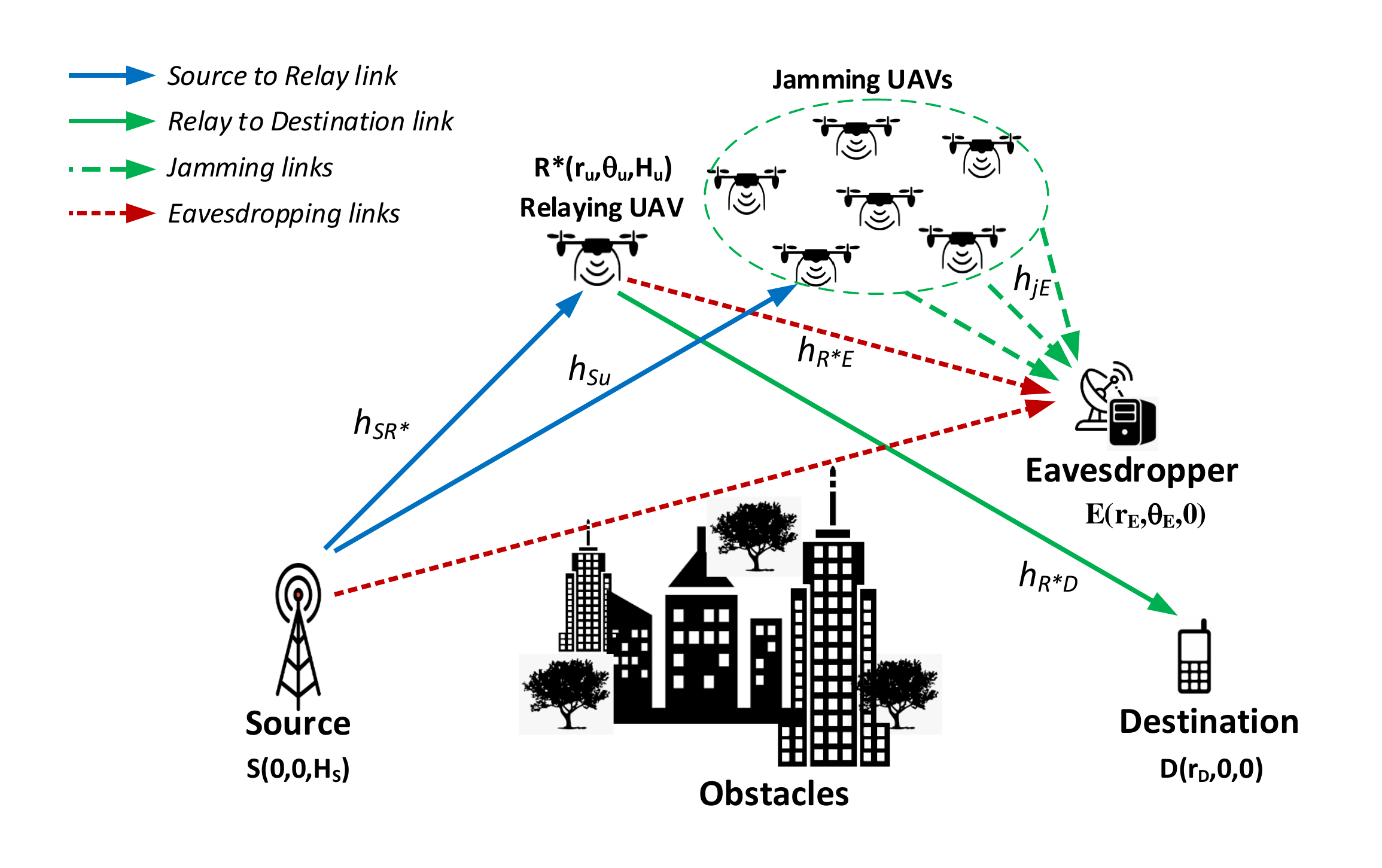}
	\caption{System model of UAV-assisted relaying network.}
	\label{fig:sys}
	\end{center}
\end{figure}
\section{System Model}
\label{Sec:SystemModel}
Consider a UAV-aided relaying system as depicted in Fig.~\ref{fig:sys}, in which a terrestrial BS $S$ communicates with an on-ground mobile user $D$, in the presence of an eavesdropper $E$ on the ground. We assume that the direct link between $S$ and $D$ is not available, e.g., due to blockages and/or long distance. For that, the communications and security of the transmission from $S$ to $D$ {are} assisted by a swarm of $U$ UAVs that {functions} as a relay and friendly jammers. Let ${R}_{u}$ denote the {\textit{u}}-th UAV where $u\in {\Phi }_{u}=\{1,2,...,U\}$. Due to their limited energy, 
we assume that UAVs are only equipped with a single antenna, operate in the half-duplex AF mode, and can wirelessly harvest power from $S$ \cite{9163290}.
Specifically, the TSR protocol between $S$ and $D$ is accomplished over three phases with the total length of $T$, i.e., the EH phase, the source to the UAVs phase, and the AF from the relay UAV to the destination {phase, as illustrated in Fig.~\ref{fig:TSRP}. Here,} we adopt the two equal time-slots AF relaying system \cite{m_tatar_mamaghani_performance_2019,b_ji_performance_2019}. During the EH phase of length $\alpha T$, $\alpha \in \left[ 0,1 \right]$ where $\alpha$ is the EH time factor, all $U$ UAVs scavenge RF energy from $S$. In the second phase of length $\left(1-\alpha \right) T/2$, $S$ broadcasts the information signal to all the UAVs, and is susceptible to eavesdropping by $E$. Note that in the second phase, all half-duplex UAVs that are in the reception mode (to receive the signal from the source) cannot harvest energy and jam the eavesdropper. Then, in third phase of length $\left(1-\alpha \right) T/2$, the UAV $R^*$ with the highest SNR over the {source-to-UAVs} links uses the harvested energy to AF the received signals to $D$, while $(U-1)$ other UAVs $R_j$ (with $j \in {\Phi _u}\backslash R^*$) use their harvested energy to jam $E$.
{In order to eliminate the burden of signal synchronization among UAVs, we consider that only one UAV is selected as a relay to forward signals to $D$.}
We assume that the friendly jamming signals can be completely canceled at $D$, e.g., using the successive cancellation or projection technique \cite{tse_viswanath_2005}. 
\begin{figure}[t!]
    \begin{center}
	\includegraphics[width=0.75\linewidth]{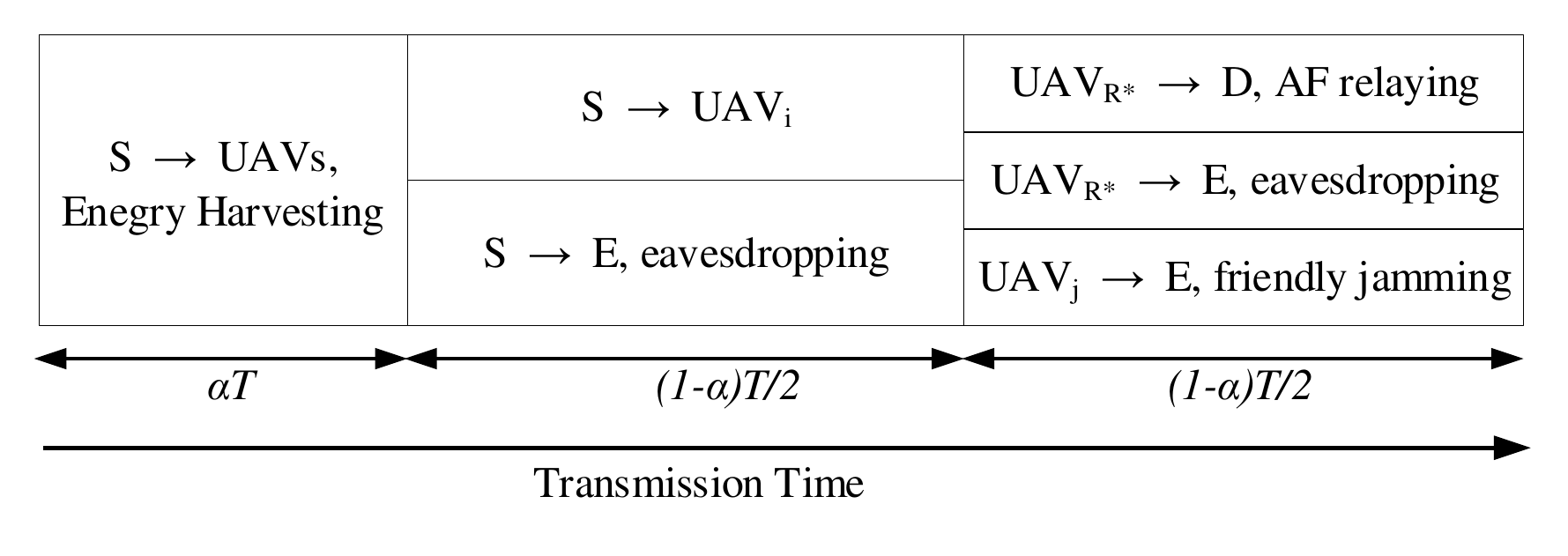}
	\setlength{\belowcaptionskip}{-8pt}
	\caption{Time-switching relaying protocol for UAV-assisted network.}
	\label{fig:TSRP}
	\end{center}
\end{figure}
\subsection{Channel Model}
All the channels are assumed to be quasi-static, i.e., unchanged during each transmission time slot $T$ but independently vary from one slot to another \cite{g._zhang_securing_2019}. The channel coefficient between nodes $u$ and $v$ is denoted as $h_{\textit{\text{uv}}}$, which has the corresponding channel gain  ${\left| h_{\textit{\text{uv}}} \right|}^{2}$. 
Specifically, ${h}_{\textit{\text{Su}}}$ is the channel coefficient between the source and {the \textit{u}}-th UAV; ${h}_{\textit{\text{SR\textsuperscript{*}}}}$ and ${h}_{\textit{\text{SE}}}$ are the channel coefficients from $S$ to $R^*$ and $E$, respectively. The channel coefficients between $R^*$ and $D$, $R^*$ and $E$, {the \textit{j}}-th jamming UAV and $E$ are ${h}_{\textit{\text{R\textsuperscript{*}D}}}$, ${h}_{\textit{\text{R\textsuperscript{*}E}}}$ and ${h}_{\textit{\text{jE}}}$, respectively, as shown in Fig.~\ref{fig:sys}. Due to the strong LoS components, all the {channels between UAVs and ground nodes} are modeled by shadowed-Rician fading\footnote{The shadowed-Rician distribution has been proposed to generally describe the {channels between UAVs and ground nodes} \cite{j_f_paris_closed-form_2010} since these channels vary significantly with UAVs' 3D locations in an area, which may be under a deep fade and shadowing.} 
with the PDF is given by 
\begin{equation}
	\begin{aligned}
		& {{f}_{{\left| {{h}_{\textit{\text{uv}}}} \right|}^{2}}}\left( x \right)=\mathcal{A}{{e}^{-\mathcal{B}x}}{}_{1}{{F}_{1}}\left( {{m}_{S}};1;\vartheta x \right),x\ge 0 ,\\ 
	\end{aligned}
\end{equation}
where $\mathcal{A}={{{\left( {2b{{m}_{S}}}/{\left( 2b{{m}_{S}}+{{\Omega }} \right)}\; \right)}^{{{m}_{S}}}}}/{2b}$, $\mathcal{B}={1}/{2b}$, $\vartheta={{{\Omega }}}/\left( 2b{{m}_{S}}+{{\Omega }} \right)/{2b}$ with ${{\Omega}}$ and $2{{b}}$ being the average power of LoS and multipath components, respectively, ${{m}_{S}}$ is the fading severity parameter, and ${}_{1}{{F}_{1}}\left( \cdot ;\cdot ;\cdot  \right)$ is the confluent hypergeometric function of the first kind \cite{zwillinger_9_2014}. For arbitrary integer-valued fading severity parameters of $m_S$, one can simplify ${}_{1}{{F}_{1}}\left( {{m}_{S}};1;\vartheta x \right)$ to obtain the PDF and CDF as \cite{v_bankey_secrecy_2017}  
\begin{equation}
	\begin{aligned}	
		&{{f}_{{\left| {{h}_{\textit{\text{uv}}}} \right|}^{2}}}\left( x \right)=\sum\limits_{l=0}^{{{m}_{S}}-1}{{{\zeta}_{\textit{\text{uv}}}}}{{x}^{l}}{{e}^{-{\eta_{\textit{\text{uv}}}}x}},
		\enspace
%
		{{F}_{{\left| {{h}_{\textit{\text{uv}}}} \right|}^{2}}}\left( x \right)=1-\sum\limits_{l=0}^{{{m}_{S}}-1}\sum\limits_{q=0}^{l}{\kappa_{\textit{\text{uv}}}}{{x}^{q}}{{e}^{-{\eta_{\textit{\text{uv}}}}x}},\\
	\end{aligned}	
\end{equation}
%
where ${{\zeta }_{\textit{\text{uv}}}}=\mathcal{A} {{{\left( {{m}_{S}}-l \right)}_{l}}{{\left( \upsilon  \right)}^{l}}}/{{{\left( l! \right)}^{2}}}$, $\kappa_{\textit{\text{uv}}}={{\zeta }_{\textit{\text{uv}}}}{\left({l!}/{q!}\right)}{{\eta_{\textit{\text{uv}}}}^{-\left( l+1-q \right)}}$, and ${\eta_{\textit{\text{uv}}}}={\mathcal{B}-\vartheta }$, in which ${{\left( a \right)}_{n}}$ is the Pochhammer symbol \cite{zwillinger_xli}. We denote $X={{\left| {{h}_{\textit{\text{SR\textsuperscript{*}}}}} \right|}^{2}}$, $Y={{\left| {{h}_{\textit{\text{R\textsuperscript{*}}D}}} \right|}^{2}}$, $Z={{\left| {{h}_{\textit{\text{R\textsuperscript{*}}E}}} \right|}^{2}}$. Therefore, $\left\{{\zeta_X},{\kappa }_{X},{\eta }_{X} \right\}$, $\left\{{\zeta_Y},{\kappa }_{Y},{\eta }_{Y} \right\}$, and $\left\{{\zeta_Z},{\kappa }_{Z},{\eta }_{Z} \right\}$ are the corresponding channel parameters of $S$-to-$R^*$, $R^*$-to-$D$, and $R^*$-to-$E$ links, respectively.

The terrestrial link between $S$ and $E$ is subject to undergo small-scale Rayleigh model due to many obstructions on ground \cite{g._zhang_securing_2019} with the PDF and CDF as  
\begin{equation}
	{{f}_{{\left| h_{\textit{\text{uv}}} \right|}^{2}}}\left( x \right)={{e}^{-{x}}}, \quad {{F}_{{\left| h_{\textit{\text{uv}}} \right|}^{2}}}\left( x \right)=1-{{e}^{-{x}}},
\end{equation} 
respectively. We denote $W={{\left| {{h}_{\textit{\text{SE}}}} \right|}^{2}}$ as the channel gain of link between $S$ and $E$.

Considering the nodes' locations, the free-space path loss (FSPL) from $u$ to $v$ is given by $\lambda_{\textit{\text{uv}}} = \left({ d_{\textit{\text{uv}}}}/{d_0}\right)^{ -\tau }$,   
	where $\tau$ is the path-loss exponent, $d_{\textit{\text{uv}}}$ is the $u$-to-$v$ distance, and $d_{0}$ is the reference distance, {e.g.,} $\lambda_{uv}\left(d_{\textit{\text{uv}}}={d_0}\right) = 1$.
The received power at $v$ can be written as
	$P_{\textit{\text{v}}} = P_u \lambda_{\textit{\text{uv}}} {\left| {{h}_{\textit{\text{uv}}}} \right|}^{2}$, where $P_u$ is the transmitted power at $u$.

We assume that the source transmits the RF signal for wireless powering and the information signal at its maximum power to improve the SNR, which is denoted as ${{P}_{S}}$. 
The noise at all the receivers is assumed to be the Additive White Gaussian Noise (AWGN) following $\mathcal{C}\mathcal{N}\left( 0,{\sigma}^2 \right)$.
\subsection{Legitimate Communications}
\label{subsec:main}
In the EH phase, following the linear EH model \cite{d_n_k_jayakody_self-energized_2020}, the harvested energy at the {\textit{u}}-th UAV is written as 
\begin{equation}
	{\Xi_{u}}=\eta \alpha T\left( {{P}_{S}}{{\lambda }_{\textit{\text{Su}}}}{{\left| {{h}_{\textit{\text{Su}}}} \right|}^{2}}+{{\sigma}^2} \right) ,
	\label{eq:EH}
\end{equation}
where $0<\eta<1$ is the energy conversion efficiency factor which depends on the EH circuitry. 
In the second phase, $S$ broadcasts its information signal to all the UAVs. The instantaneous SNR over the link between $S$ and {the \textit{u}-th UAV} is ${{\gamma }_{{\textit{\text{Su}}}}}={\psi}{\lambda_{\textit{\text{Su}}}}{{\left| {{h}_{\textit{\text{Su}}}} \right|}^{2}}$,
with denotation of $\psi ={{P}_{S}}/{{\sigma}^2}$. 
At the end of this phase, only $R^*$ is selected as a relay node according to the highest SNR criterion as ${\gamma }_{\textit{\text{SR\textsuperscript{*}}}} = {\psi}{\lambda_{\textit{\text{SR\textsuperscript{*}}}}} {{\left| {{h}_{\textit{\text{SR\textsuperscript{*}}}}} \right|}^{2}}=\underset{u\in {{\Phi }_{u}}}{{\max }}\,\left( {\psi}{\lambda_{\textit{\text{Su}}}}{{\left| {{h}_{\textit{\text{Su}}}} \right|}^{2}}\right)$.
%
The harvested energy at $R^*$ in the EH phase is ${\Xi_{R^*}}=\eta \alpha T\left( {{P}_{S}}{{\lambda }_{{\textit{\text{SR}}}^*}}{{\left| {{h}_{{\textit{\text{SR}}}^*}} \right|}^{2}}+{{\sigma}^2} \right)$.
Then, in the third phase, $R^*$ uses its harvested energy to AF the received signals to $D$ with an amplification factor of 
\begin{equation}
	{G_\textit{AF}} =\frac{2{\Xi_{R^*}}/\left(1-\alpha \right) T}{{P}_{S}{\lambda_{\textit{\text{SR\textsuperscript{*}}}}}{{\left| {{h}_{\textit{\text{SR\textsuperscript{*}}}}} \right|}^{2}}+{{\sigma}^2}}.
	\label{eq:amp_def}
\end{equation}
The SNR of received signal at $D$ is then written as
\begin{equation}
	{{\gamma }_{D}}=\frac{\varepsilon\psi {\lambda_{\textit{\text{SR\textsuperscript{*}}}}}{\lambda_{\textit{\text{R\textsuperscript{*}}D}}} XY}{\varepsilon {\lambda_{\textit{\text{R\textsuperscript{*}}D}}} {Y}+1},
	\label{eq:snrDfinal}
\end{equation}
where $\varepsilon=
2{\eta\alpha}/{\left( 1-\alpha \right)}$ denotes the time factor of wireless-powered TSR protocol. 
\subsection{Eavesdropping}
\label{subsec:wire_tap}
$E$ is assumed to attempt to eavesdrop information in the phase 2 and phase 3. 
In the second phase, $E$ directly listens to $S$ and receives the signal with the SNR as
\begin{equation}
	{{\gamma}_{\textit{\text{SE}}}}=\psi {\lambda_{\textit{\text{SE}}}} W .
	\label{eq:snrSE}
\end{equation}
%
In the third phase, without friendly jamming the SNR of received signal at $E$ has the same form as ${\gamma}_{D}$ in \eqref{eq:snrDfinal} 
\begin{equation}
{{\gamma }_{\textit{\text{R\textsuperscript{*}}E}}}=\frac{\varepsilon\psi {\lambda_{\textit{\text{SR\textsuperscript{*}}}}}{\lambda_{\textit{\text{R\textsuperscript{*}}E}}} XZ}{\varepsilon {\lambda_{\textit{\text{R\textsuperscript{*}}E}}} {Z}+1}.
\label{eq:snrREfinal}
\end{equation}
%
$E$ can intelligently either perform SC to select the highest SNR received signal as ${{\gamma}_{E}^{{{\text{SC}}}}}=\max \left( {{\gamma }_{\textit{\text{SE}}}},{{\gamma }_{\textit{\text{R\textsuperscript{*}}E}}} \right)$, or MRC with the sum SNRs as ${{\gamma}_{E}^{{\text{MRC}}}}={{\gamma}_{\textit{\text{SE}}}}+{{\gamma}_{\textit{\text{R\textsuperscript{*}}E}}}$ to intercept the legitimate information \cite{x._yuan_secrecy_2020}. 
\subsection{Cooperative Jamming}
When all the UAVs other than the relaying one use their harvested energy to send jamming signals to $E$, the SNR at $E$ over aerial links is 
\begin{equation}
	{{\gamma }_{\textit{\text{R\textsuperscript{*}E}}}^\text{J}} = \frac{\varepsilon\psi {\lambda_{\textit{\text{SR\textsuperscript{*}}}}}{\lambda_{\textit{\text{R\textsuperscript{*}}E}}}XZ}{\varepsilon {\lambda_{\textit{\text{R\textsuperscript{*}}E}}} Z+1+{{P}_\textit{\text{JE}}}},
	\label{eq:SNR_RJE}
\end{equation}
where ${{P}_\textit{\text{JE}}}=\delta \varepsilon \sum\limits_{j=1}^{U-1}{\left\{\psi {{{\lambda }}_{\textit{\text{Sj}}}{\left| {{h}_{Sj}} \right|}^{2}}+1\right\}}{{{{\lambda }}_{\textit{\text{jE}}}}{\left| {{h}_{jE}} \right|}^{2}}$, and $\delta$ is the factor of using harvested energy, $0<\delta \le 1$.
We assumed that all {the channels between UAVs and ground nodes} are i.i.d. with ${{\lambda }}_{\textit{\text{Sj}}}={{\lambda }}_{\textit{\text{SR\textsuperscript{*}}}}$ and ${{{\lambda }}_{\textit{\text{jE}}}}={{{\lambda }}_{\textit{\text{R\textsuperscript{*}E}}}}$, $j \in {\Phi _u}\backslash R^*$, the jamming power from each {of the jamming UAVs} can be approximated as a fraction $\delta$ of the average energy harvested from $S$, 
$	P_J=\delta \varepsilon {\psi {{\lambda_{\textit{\text{SR\textsuperscript{*}}}}}g}}$, 
where $g$ is the average channel gain between $S$ and jamming UAVs. 
Denote ${\mathcal{J}}=\sum\limits_{j=1}^{U-1}{{{\left| {{h}_{\textit{\text{jE}}}} \right|}^{2}}}$, we obtain 
\begin{equation}
	{{\gamma }_{\textit{\text{R\textsuperscript{*}E}}}^\text{J}} = \frac{\varepsilon\psi {\lambda_{\textit{\text{SR\textsuperscript{*}}}}}{\lambda_{\textit{\text{R\textsuperscript{*}}E}}}XZ}{\varepsilon {\lambda_{\textit{\text{R\textsuperscript{*}}E}}} Z+1+{P_J\lambda_{\textit{\text{R\textsuperscript{*}E}}}}{\mathcal{J}}},
	\label{eq:SNR_RJ3}
\end{equation}

\subsection{Distributions of Defined Random Variables $X$ and ${\mathcal{J}}$}
Under the assumption that the UAVs are located sufficiently apart within a swarm whose span is very small as compared to the distances from UAVs to on-ground nodes, we assume i.i.d. {channels between $S$ and UAVs} with the same average received power ${{{\lambda }}_{\textit{\text{Su}}}}={{{\lambda }}_{\textit{\text{SR\textsuperscript{*}}}}}$. 
Therefore, the UAV with the
highest channel gain is selected to relay information to $D$. 
The CDF of $X={{\left| {{h}_{\textit{\text{SR\textsuperscript{*}}}}} \right|}^{2}}=\underset{u\in {{\Phi }_{u}}}{{\max }}\, {{\left| {{h}_{\textit{\text{S}u}}} \right|}^{2}} $ is written as follows:
\begin{equation}
	{{F}_{X}}\left( x \right)=\Pr \left\{ \left. {{\left| {{h}_{\textit{\text{Su}}}} \right|}^{2}}<x \right|u\in {{\Phi }_{u}} \right\} {\mathop{=}}\,\prod\limits_{u=1}^{U}{{{F}_{{\left| {{h}_{\textit{\text{Su}}}} \right|}^{2}}}\left( x \right)}.
	\label{eq:CDFX1}
\end{equation}
\begin{lemma}
	The expressions for the CDF and PDF of r.v. $X$ are presented as 
	\begin{equation}
		{{F}_{X}}\left( x \right)={\widetilde{\sum }_{u}^U}{{\kappa }_{u}}{{x}^{{{\chi }_{u}}}}{{e}^{-{{\eta }_{u}}x}},		
		\label{eq:CDFX}
	\end{equation}
	\begin{equation}
		\begin{aligned}
			&{{f}_{X}}\left( x \right)=U\sum\limits_{l_X=0}^{{{m}_{S}}-1}{{\zeta }_{X}}{\widetilde{\sum }_{u}^{{U-1}}}{{\kappa }_{u}}{{x}^{{l}_{X}+{\chi }_{u}}}{{e}^{-\left({\eta }_{X}+{\eta }_{u}\right)x}},
			\label{eq:PDFX}
		\end{aligned}
	\end{equation}
	where
	\begin{equation*}
		\begin{aligned}
			& {\widetilde{\sum }_{u}^U}=\sum\limits_{u=0}^{U}{\frac{{{\left( -1 \right)}^{u}}}{u!}}\underbrace{\sum\limits_{{{n}_{1}}=1}^{U}{\ldots }\sum\limits_{{{n}_{u}}=1}^{U}{}}_{{{n}_{1}}\ne {{n}_{2}}\ldots \ne {{n}_{u}}}\ \sum\limits_{{{l}_{1}}=0}^{{{m}_{S}}-1}{\sum\limits_{{{q}_{1}}=0}^{{{l}_{1}}}{\ldots }}\sum\limits_{{{l}_{u}}=0}^{{{m}_{S}}-1}{\sum\limits_{{{q}_{u}}=0}^{{{l}_{u}}}} ,\\
			&
			{\kappa}_{u}=\prod\limits_{t=1}^{u}{{\kappa }_{X,t}}, \quad {{\eta }_{u}}=\sum\limits_{t=1}^{u}{\eta_{X,t}}, \quad{{\chi }_{u}}=\sum\limits_{t=1}^{u}{{{q}_{t}}}
			. \\
		\end{aligned}
		\label{eq:new v.r}
	\end{equation*}
	\label{lemma:FX}
\end{lemma}
\begin{proof}
	Using the multinomial theorem, the CDF of r.v. $X$ is given as in \eqref{eq:CDFX}. The corresponding PDF in \eqref{eq:PDFX} can be obtained by taking the derivative of ${F_X}\left( x \right)$ with respect to (w.r.t.) $x$ to complete the proof. 
\end{proof}
We use the Moment Generating Function (MGF) approach to derive the PDF of ${{\mathcal{J}}}=\sum\limits_{j=1}^{U-1}{{{\left| {{h}_{jE}} \right|}^{2}}}$.
\begin{lemma}
	The expression for the PDF of ${\mathcal{J}}$ is presented as 	
\begin{equation}
	{{f}_{{\mathcal{J}}}}\left( t \right)=\widehat{\sum}_{k}^{U-1}{\zeta_k}{{t}^{\chi _{k}-1}}{{e}^{-{{\eta }_{Z}}t}},\enspace U \ge 1,
	\label{eq:PDF_Zsum}
\end{equation}
where	 
\begin{equation*}
	\begin{aligned}
		&\widehat{\sum}_{k}^{U-1}=\sum\limits_{{{k}_{1}}=0}^{{{m}_{S}}-1}{\ldots} \sum\limits_{{{k}_\text{U-1}}=0}^{{{m}_{S}}-1}, \enspace 
		\chi_{k}=\sum\limits_{j=1}^{U-1}{\left({k}_{j}+1\right)},\enspace {\zeta_k}=\frac{1}{\left( \chi _{k}-1 \right)!}{{\prod\limits_{j=1}^{U-1}{\left( {{\zeta }_{Z}}{{k}_{j}}! \right)} }}. \\
	\end{aligned}
	\label{eq:new v.r3}
\end{equation*}
\end{lemma}
\begin{proof}
	The proof is provided in Appendix A.
\end{proof}

\section{Secrecy Performance Analysis}
\label{Sec:Secrecy}
Secrecy capacity is defined as the positive value of the difference between the instantaneous capacities of the legitimate and the wiretap channels as
${{{C}_{S}}=\left( {{C}_{D}}-{{C}_{E}}\right)^{+}}$  \cite{a_d_wyner_wire_tap_1975}. 
To measure the security performance of the system, the SOP is defined as the probability at which the achieved secrecy capacity is not greater than a predefined secrecy rate of $C_{\textit{\text{th}}}$. Let $P_{\textit{\text{out}}}\left({C }_\textit{\text{th}}\right)=\Pr \left\{ {{C}_{S}}<{{C }_{th}} \right\}$ denote the SOP at $C_{\textit{\text{th}}}$ 
\begin{equation}
	\begin{aligned}
P_{\textit{\text{out}}}\left({C }_\textit{\text{th}}\right)
&=\Pr \left\{\frac{1-\alpha }{2}{{\log }_{2}}\left(\frac{1+{{\gamma }_{D}}}{1+{{\gamma }_{E}}}\right)<{{C }_\textit{\text{th}}} \right\} 
=\Pr \left\{ \frac{{\gamma }_{D}}{{\gamma }_{S}} < {{\gamma }_{E}+1-\frac{1}{\gamma_S}} \right\} ,\\
	\end{aligned}
	\label{eq:SOP_def}
\end{equation}
where $\gamma_{S}=2^{2\mathcal{C}_{\textit{\text{th}}}/(1-\alpha)}$ {denotes} the target secrecy SNR.
In high SNR regime, exploiting the approximation of $\frac{1+{{\gamma }_{D}}}{1+{{\gamma }_{E}}}\approx \frac{{{\gamma }_{D}}}{{{\gamma }_{E}}}$, which is widely adopted in literature \cite{v_bankey_secrecy_2017}, we obtain the asymptotic expression of ${{P}_{\textit{\text{out}}}}$ as
\begin{equation}
{{P}_{\textit{\text{out}}}}\left( {{\gamma }_{S}} \right)= \Pr \left\{ {\frac{{{\gamma }_{D}}}{{{\gamma }_{E}}}}<{{\gamma }_{S}}\right\}.
\end{equation}
%
Since there are many factors in our model, the SOP cannot be directly analyzed for the case of a randomly distributed eavesdropper. Therefore, in this section we derive the expressions of SOP with assumption that the CSI of $E$ is available to the network. 
%

\subsection{Selection Combining}
The asymptotic SOP for the SC scheme at $E$ is
\begin{equation}
	\begin{aligned}
		P_{\textit{\text{out}}}^{{\text{SC}}}\left( {{\gamma }_{S}} \right)&=1-\Pr \left\{ \frac{{{\gamma }_{D}}}{{{\gamma }_{S}}}>\max \left( {{\gamma }_{\textit{\text{SE}}}},{{\gamma }_{\textit{\text{R\textsuperscript{*}}E}}} \right) \right\} .\\ 
	\end{aligned}
	\label{eq:SOP_SC}
\end{equation}
Using the formulas of SNRs in \eqref{eq:snrDfinal}, \eqref{eq:snrSE} and \eqref{eq:snrREfinal}, we get
\begin{equation}
	\begin{split}
		P_{\textit{\text{out}}}^{{\text{SC}}}\left( {{\gamma }_{S}} \right)&=1-{{\mathbb{E}}_{X,Y}}\left\{ \Pr \left\{ \begin{aligned}
			& W<{{\Upsilon_{W}}(Y)}X, \enspace Z<{{\Upsilon_Z}(Y)} \\ 
		\end{aligned} \right\} \right\},		
	\end{split}				
	\label{eq:SOP_SC1}
\end{equation} 
where 
\begin{equation*}
	\begin{split}
		&{{\Upsilon_{W}}(Y)}=\frac{\varepsilon {\lambda_{\textit{\text{SR\textsuperscript{*}}}}} {\lambda_{\textit{\text{R\textsuperscript{*}}D}}}Y}
		{{{\gamma }_{S}}{\lambda_{\textit{\text{SE}}}}\left( \varepsilon {\lambda_{\textit{\text{R\textsuperscript{*}}D}}}Y+1 \right)} ,
		\enspace {{\Upsilon_Z}(Y)}=\frac{{\lambda_{\textit{\text{R\textsuperscript{*}}D}}}Y}
		{{\lambda_{\textit{\text{R\textsuperscript{*}}E}}}\left( \varepsilon {\lambda_{\textit{\text{R\textsuperscript{*}}D}}}\left( {{\gamma }_{S}}-1 \right)Y+{{\gamma }_{S}} \right)}.\\
	\end{split}
\end{equation*} 
\begin{figure*}[!tb]
	\begin{center}
		\begin{equation}
			\begin{aligned}
				P_{\textit{\text{out}}}^{{\text{SC}}}
				\left( {{\gamma }_{S}} \right)=1
				-& U				 \sum\limits_{{{l}_{X}}=0}^{{{m}_{S}}-1}
				{\widetilde{\sum }_{u}^{U-1}\sum\limits_{{{l}_{Y}}=0}^{{{m}_{S}}-1}{{{\zeta }_{X}}{{\zeta }_{Y}}}}{{\kappa }_{u}}\Gamma \left( {{\chi }_{u}}+1 \right) 
				\Biggl\{ 
				{{\eta }_{u}}^{-\left( {{\chi }_{u}}+1 \right)}\Gamma \left( {{l}_{Y}}+1 \right){{\eta }_{Y}}^{-\left( {{l}_{Y}}+1 \right)} \\
				& \quad -{{\left( {{{\tilde{\gamma }}}_{SE}} \right)}^{{{\chi }_{u}}+1}}
			    {\mathcal{I}_1}
				-\sum\limits_{{{l}_{Z}}=0}^{{{m}_{S}}-1}{\sum\limits_{{{q}_{Z}}=0}^{{{l}_{Z}}}{{{\kappa }_{Z}}}} {{\left( {{\tilde{\gamma }}_{Z}} \right)}^{{{q}_{Z}}}} 
				\left\{ {{\eta }_{u}}^{-\left( {{\chi }_{u}}+1 \right)} {\mathcal{I}_2} 
				+ {{\left( {{{\tilde{\gamma }}}_{SE}} \right)}^{{{\chi }_{u}}+1}} {\mathcal{I}_3}  \right\} \Biggr\} . \\ 				
			\end{aligned} 
			\label{eq:SOP_SC_final}
		\end{equation}
	\end{center}
	\setlength{\arraycolsep}{1pt}
	\hrulefill \setlength{\arraycolsep}{0.0em}
\end{figure*}
\begin{figure*}[!tb]
\begin{center}
	\begin{equation}
		\begin{split}
			&{{{\Theta }_{1}}\left( v,\gamma ,\mu ;\alpha ,\beta  \right)={{e}^{\mu \beta }}\sum\limits_{m=0}^{v}{\left( \begin{aligned}
						& v \\ 
						& m \\ 
					\end{aligned} \right)}{{\left( -\beta  \right)}^{v-m}}\sum\limits_{n=0}^{\gamma }{\left( \begin{aligned}
						& \gamma  \\ 
						& n \\ 
					\end{aligned} \right)}{{\left( \alpha -\beta  \right)}^{\gamma -n}}{{\Phi }_{1}}\left( \beta ;m+n-\gamma ,\mu  \right)} .\\				
		\end{split}	
		\label{eq:Theta_1}
	\end{equation}
	\setlength{\arraycolsep}{1pt}
	\hrulefill \setlength{\arraycolsep}{0.0em}
\end{center}
\end{figure*}
\begin{figure*}[!tb]
\begin{center}
	\begin{equation}
		\begin{split}
			&{{{\Theta }_{2}}\left( v,\gamma ,\mu ,\rho ;\beta  \right)={{e}^{\mu \beta -\rho }}\sum\limits_{m=0}^{v+\gamma }{\left( \begin{aligned}
						& v \\ 
						& m \\ 
					\end{aligned} \right)}{{\left( -\beta  \right)}^{v+\gamma -m}}\sum\limits_{n=0}^{\infty }{\frac{{{\left( \rho \beta  \right)}^{n}}}{n!}} {{\Phi }_{1}}\left( \beta ;m-n-\gamma ,\mu  \right)} .\\
		\end{split}	
		\label{eq:Theta_2}
	\end{equation}
	\setlength{\arraycolsep}{1pt}
	\hrulefill \setlength{\arraycolsep}{0.0em}	
\end{center}
\end{figure*}
\begin{figure*}[!tb]
\begin{center}
	\begin{equation}
		\begin{split}
			& {\begin{aligned}
					{{\Theta }_{3}}
					&\left( v,\gamma ,\lambda ,\mu ,\rho ;\alpha ,\beta ,\xi  \right) = \sum\limits_{p=0}^{\infty }{\frac{{{\left( -\rho  \right)}^{p}}}{p!}} \sum\limits_{m=0}^{\gamma }
					{\left( \begin{aligned}
							& \gamma  \\ 
							& m \\ 
						\end{aligned} \right)}
					{{\alpha }^{\gamma -m}} \\
					& \quad\quad\quad \times \left\{
					\begin{aligned}					\sum\limits_{a=1}^{\gamma +p}{{A}_{\gamma +p-a+1}}{{\Phi }_{2}} \left( \gamma +p-a+1,\mu ;\beta  \right) + \sum\limits_{b=1}^{\lambda }{{{B}_{\lambda -b+1}}{{\Phi }_{2}}\left( \lambda -b+1,\mu ;\xi  \right)}
					\end{aligned}
					 \right\} .\\ 
			\end{aligned} }\\
			& {\begin{aligned}
					& \begin{split}
						{{A}_{\gamma +p-a+1}} = 
						& \frac{1}{\left( a-1 \right)!}
						\sum\limits_{q=0}^{a-1}{\left( \begin{aligned}
								& a-1 \\ 
								& \ \ q \\ 
							\end{aligned} \right)}{{\left( -1 \right)}^{q}}{{\left( \xi -\beta  \right)}^{-\left( \lambda +q \right)}} \\
						& \quad \times \prod\limits_{r=1}^{q}{\left( \lambda +r-1 \right)}{{\left( -\beta  \right)}^{v+\lambda +m+p-\left( a-1-q \right)}}\prod\limits_{s=1}^{a-1-q}{\left( v+\lambda +m+p-s+1 \right)} .\\
					\end{split} \\
					& \begin{split}
						{{B}_{\lambda -b+1}}= 
						& \frac{1}{\left( b-1 \right)!}
						\sum\limits_{q=0}^{b-1}{\left( \begin{aligned}
								& b-1 \\ 
								& \ \ q \\ 
							\end{aligned} \right)}{{\left( -1 \right)}^{q}}{{\left( \beta -\xi  \right)}^{-\left( \gamma +p+q \right)}} \\
						& \quad \times \prod\limits_{r=1}^{q}{\left( \gamma +p+r-1 \right)}{{\left( -\xi  \right)}^{v+\lambda +m+p-\left( b-1-q \right)}}\prod\limits_{s=1}^{b-1-q}{\left( v+\lambda +m+p-s+1 \right)} .\\
					\end{split} \\ 
			\end{aligned}}\\				
		\end{split}	
		\label{eq:Theta_3}
	\end{equation}
\end{center}
\setlength{\arraycolsep}{1pt}
\hrulefill 
\setlength{\arraycolsep}{0.0em}
\end{figure*}
\begin{figure*}[!tb]
	\begin{center}
		\begin{equation}
			\begin{aligned}
				&{{{\Phi }_{1}}\left( u;v,\mu  \right)=\left\{ \begin{aligned}
						& {{e}^{-u\mu }}\sum\limits_{k=0}^{v}{\frac{v!}{k!}}\frac{{{u}^{k}}}{{{\mu }^{v-k+1}}},\quad \mu >0, v\ge 0 ,\\ 
						& \begin{aligned} 
						    {{\left( -1 \right)}^{-v}} 
						    & \frac{{{\mu }^{-v-1}}\text{Ei}\left( -\mu u \right)}{\left( -v-1 \right)!} \\
						    & +\frac{{{e}^{-\mu u}}}{{{u}^{-v-1}}} \sum\limits_{k=0}^{-v}{\frac{{{\left( -1 \right)}^{k}}{{\mu }^{k}}{{u}^{k}}}{\left( -v-1 \right)\left( -v-1-1 \right)\ldots \left( -v-1-k \right)}},\quad \mu >0, v<0 .\\
						\end{aligned} \\ 
					\end{aligned} \right.} \\
				& {{\Phi }_{2}}\left( \gamma ,\mu ;\beta  \right)=\frac{{{\left( -\mu  \right)}^{\gamma -1}}}{\left( \gamma -1 \right)!}\left( \sum\limits_{g=1}^{\gamma -1}{\frac{\left( g-1 \right)!}{{{\left( -\beta \mu  \right)}^{g}}}}-{{e}^{\beta \mu }}\text{Ei}\left( -\beta \mu  \right) \right) .\\
			\end{aligned}
			\label{eq:Phi}
		\end{equation}	
	\end{center}
	\setlength{\arraycolsep}{1pt}
	\hrulefill \setlength{\arraycolsep}{0.0em}
\end{figure*}
\begin{proposition}
	Without friendly jamming, the asymptotic SOP with respect to SC scheme at $E$ is presented in \eqref{eq:SOP_SC_final}, where
	\begin{equation*}
		\begin{aligned}
			&{\mathcal{I}_1}= {{\Theta }_{1}}\left( {{l}_{Y}},{{\chi }_{u}}+1,{{\eta }_{Y}};\frac{1}{\varepsilon {\lambda_{\textit{\text{R\textsuperscript{*}}D}}}},{{{\tilde{\eta }}}_{u}} \right) ,\enspace {\mathcal{I}_2}={{\Theta }_{2}}\left( {{l}_{Y}},{{q}_{Z}},{{\eta }_{Y}},{{\eta }_{Z}}{{\tilde{\gamma }}_{Z}};
			{{\gamma }_{S}}{{\tilde{\gamma }}_{Y}} \right) ,\\
			&{\mathcal{I}_3}={{\Theta }_{3}}\left( {{l}_{Y}},{{\chi }_{u}}+1,{{q}_{Z}},{{\eta }_{Y}},{{\eta }_{Z}}{{\tilde{\gamma }}_{Z}};
			\frac{1}{\varepsilon {\lambda_{\textit{\text{R\textsuperscript{*}}D}}}},{{\tilde{\eta }}_{u}},{{\gamma }_{S}}{{\tilde{\gamma }}_{Y}} \right) ,\\
		\end{aligned}
	\end{equation*}
in which defined integrals are expressed in \eqref{eq:Theta_1}, \eqref{eq:Theta_2}, \eqref{eq:Theta_3}, \eqref{eq:Phi}, and 
	${{\tilde{\gamma }}_{SE}}=\frac{{{\gamma }_{S}}{\lambda_{\textit{\text{SE}}}}}
	{{{\eta }_{u}}{{\gamma }_{S}}{\lambda_{\textit{\text{SE}}}}+{\lambda_{\textit{\text{SR\textsuperscript{*}}}}}}$, \enspace
	${{\tilde{\eta }}_{u}}=\frac{{{\eta }_{u}}{{\tilde{\gamma }}_{SE}}}
	{\varepsilon{\lambda_{\textit{\text{R\textsuperscript{*}}D}}}}$, \enspace
	${{\tilde{\gamma }}_{Y}}=\frac{1}{\varepsilon {\lambda_{\textit{\text{R\textsuperscript{*}}D}}}\left( {{\gamma }_{S}}-1 \right)}$, \enspace
	${{\tilde{\gamma }}_{Z}}=\frac{1}{\varepsilon {\lambda_{\textit{\text{R\textsuperscript{*}}E}}}\left( {{\gamma }_{S}}-1 \right)}$. 

\begin{proof}
	The proof is provided in Appendix B.
\end{proof}
\end{proposition}
In case of cooperative jamming, ${{\gamma }_{\textit{\text{R\textsuperscript{*}}E}}}$ in \eqref{eq:SOP_SC} is replaced by ${{\gamma }_{\textit{\text{R\textsuperscript{*}}E}}^{\text{J}}}$. Therefore, with friendly jamming, the SOP is written as 
\begin{equation}
	\begin{aligned}
		P_{\textit{\text{out}}}^{{\text{SC,J}}}\left( {{\gamma }_{S}} \right)=1-{{\mathbb{E}}_{X,{\mathcal{J}}}}\left\{ \Pr \left\{ 
		\begin{aligned}
			& W<X{{\Upsilon }_{W}}\left( Y \right), \enspace Y>{{\Upsilon }_{Y}}\left( Z,{\mathcal{J}} \right) ,\enspace Z<{{\Upsilon }_{Z}}\left({\mathcal{J}} \right) \\ 
		\end{aligned} \right\} \right\}	,	
	\end{aligned}		
	\label{eq:SOP_SCJ1}
\end{equation}
where 
\begin{equation*}
	\begin{aligned}
		& {{\Upsilon }_{Y}}\left( Z,{\mathcal{J}} \right)=\frac{{{\gamma }_{S}}{\lambda_{\textit{\text{R\textsuperscript{*}}E}}}}{\lambda_{\textit{\text{R\textsuperscript{*}}D}}}
		\frac{Z}
		{1+{P_J{\lambda_{\textit{\text{R\textsuperscript{*}}E}}}}{\mathcal{J}}-\varepsilon {\lambda_{\textit{\text{R\textsuperscript{*}}E}}}\left( {{\gamma }_{S}}-1 \right)Z} , 
		\enspace {{\Upsilon }_{Z}}\left({\mathcal{J}} \right)=\frac{1+{P_J}{\lambda_{\textit{\text{R\textsuperscript{*}}E}}}{\mathcal{J}}}
		{\varepsilon {\lambda_{\textit{\text{R\textsuperscript{*}}E}}}\left( {{\gamma }_{S}}-1 \right)} .\\ 
	\end{aligned}
\end{equation*}
\begin{figure*}[!tb]
	\begin{center}
		\begin{equation}
			\begin{aligned}
				P_{\textit{\text{out}}}^{{\text{SC,J}}}\left( {{\gamma }_{S}} \right)=1
				- & U			 \sum\limits_{{{l}_{X}}=0}^{{{m}_{S}}-1}{{}} \widetilde{\sum }_{u}^{U-1} \sum\limits_{{{l}_{Z}}=0}^{{{m}_{S}}-1}{{\zeta }_{X}}{{\zeta }_{Z}}{{\kappa }_{u}}\Gamma \left( {{\chi }_{u}}+1 \right)
				\widehat{\sum }_{k}^{U-1}{{\zeta }_{k}} \\ 
				& \times \left\{ 
				{ 	\sum\limits_{{{l}_{Y}}=0}^{{{m}_{S}}-1}	\sum\limits_{{{q}_{Y}}=0}^{{{l}_{Y}}}{{{\kappa }_{Y}}}}
					{{\left( {{\eta }_{u}} \right)}^{-\left( {{\chi }_{u}}+1 \right)}}
					{{\left( {{\gamma }_{S}}{{\tilde{\gamma }}_{Y}} \right)}^{{{q}_{Y}}}}			  {\mathcal{I}_4}  
				- \sum\limits_{{{l}_{Y}}=0}^{{{m}_{S}}-1}{{\zeta }_{Y}}{{\left( {{{\tilde{\gamma }}}_{S}} \right)}^{{{\chi }_{u}}+1}}					 {\mathcal{I}_5}  
				\right\} .\\ 				
			\end{aligned}		
			\label{eq:SOP_SCJ_final}	
		\end{equation}
	\end{center}
	\setlength{\arraycolsep}{1pt}
	\hrulefill \setlength{\arraycolsep}{0.0em}
\end{figure*}
\begin{figure*}[!tb]
	\begin{center}
		\begin{equation}
			\begin{split}		
				&{\begin{aligned}	
						{{\Theta }_{4}}\left( u;v,\gamma ,\mu ,\rho ;\beta  \right)=
						& {{e}^{\rho -\mu \beta }}
						\sum\limits_{m=0}^{v+\gamma }{\left( \begin{aligned}
								& v+\gamma  \\ 
								& m \\ 
							\end{aligned} \right)}{{\left( -1 \right)}^{m}}{{\beta }^{v+\gamma -m}}\sum\limits_{n=0}^{\infty }{\frac{{{\left( \mu  \right)}^{n}}}{n!}\sum\limits_{p=0}^{\infty }{\frac{{{\left( -1 \right)}^{p}}{{\left( \rho \beta  \right)}^{p}}}{p!}}} \\ 
						& \times \left\{ \begin{aligned}
							& 1,\quad \gamma +m+n-p+1=0 \\ 
							& \frac{{{\beta }^{\gamma +m+n-p+1}}-{{\left( \beta -u \right)}^{\gamma +m+n-p+1}}}{\gamma +m+n-p+1} ,\enspace\gamma +m+n-p+1\ne 0 .\\ 
						\end{aligned} \right. \\ 
				\end{aligned}} \\
			\end{split}	
			\label{eq:Theta_4}
		\end{equation}
		\setlength{\arraycolsep}{1pt}
		\hrulefill \setlength{\arraycolsep}{0.0em}
		\begin{equation}
			\begin{split}		
				&{{{\Theta }_{5}}\left( u;v,\gamma ,\mu ;\alpha ,\beta  \right)={{e}^{\mu \beta }}\sum\limits_{m=0}^{v}{\left( \begin{aligned}
							& v \\ 
							& m \\ 
						\end{aligned} \right)}{{\left( -\beta  \right)}^{v-m}}\sum\limits_{n=0}^{\gamma }{\left( \begin{aligned}
							& \gamma  \\ 
							& n \\ 
						\end{aligned} \right)}{{\left( \alpha -\beta  \right)}^{\gamma -n}}{{\Phi }_{1}}\left( u+\beta ;m+n-\gamma ,\mu  \right)
				} .\\
			\end{split}	
			\label{eq:Theta_5}
		\end{equation}
		\setlength{\arraycolsep}{1pt}
		\hrulefill \setlength{\arraycolsep}{0.0em}
	\end{center}
\end{figure*}
\begin{proposition}
	With friendly jamming, the asymptotic SOP with respect to SC scheme at $E$ is presented in \eqref{eq:SOP_SCJ_final}, where 
	\begin{equation*}	
		\begin{aligned}
		& {\mathcal{I}_4}= 
			\int_{0}^{\infty }{{{t}^{\chi _{k}^{M-1}-1}}{{e}^{-{{\eta }_{Z}}t}}} 
			{{\Theta }_{4}}\left( {{\Upsilon }_{Z}}\left( t \right);{{l}_{Z}},{{q}_{Y}},{{\eta }_{Z}},{{\gamma }_{S}}{{\eta }_{Y}}{{{\tilde{\gamma }}}_{Y}};{{\Upsilon }_{Z}}\left( t \right) \right)dt , \\
		& {\mathcal{I}_5}=  \int_{0}^{\infty }{{{t}^{\chi _{k}^{M-1}-1}}{{e}^{-{{\eta }_{Z}}t}}}
			\int_{0}^{{{\Upsilon }_{Z}}\left( t \right)}{{{z}^{{{l}_{Z}}}}{{e}^{-{{\eta }_{Z}}z}}} 
			{{\Theta }_{5}}\left( {{\Upsilon }_{Y}}\left( z,t \right);{{l}_{Y}},{{\chi }_{u}}+1,{{\eta }_{Y}};\frac{1}{\varepsilon {\lambda_{\textit{\text{R\textsuperscript{*}}D}}}},{{{\tilde{\eta }}}_{u}} \right)dzdt ,
		\end{aligned}
	\end{equation*}%
	in which defined integrals are expressed in \eqref{eq:Theta_4} and 
	\eqref{eq:Theta_5}.
\end{proposition}
\begin{proof}
The proof is provided in Appendix C.
\end{proof}
\subsection{Maximum-Ratio Combining}
For the case of $E$ using the MRC scheme to increase intercepting level, the asymptotic SOP is 
\begin{equation}
	\begin{aligned}
		P_{\textit{\text{out}}}^{{\text{MRC}}}\left( {{\gamma }_\textit{\text{S}}} \right)&=\Pr \left\{\frac{{{\gamma }_\textit{\text{D}}}}{{{\gamma }_{\text{S}}}} < {{\gamma }_{\textit{\text{SE}}}}+{{\gamma }_{\textit{\text{R\textsuperscript{*}E}}}} \right\} .\\ 
	\end{aligned}
	\label{eq:SOP_MRC}
\end{equation}
Using the formulas of SNRs in \eqref{eq:snrDfinal}, \eqref{eq:snrSE}, and \eqref{eq:snrREfinal}, we get
\begin{equation}
	\begin{aligned}
		P_{\textit{\text{out}}}^{{\text{MRC}}}\left( {{\gamma }_{S}} \right)
		&={{\mathbb{E}}_{X,Y}}\left\{
		\Pr \left\{ Z>{{\Upsilon_Z}(Y)} \right\} +\Pr \left\{ 
			 Z<{{\Upsilon_Z}(Y)}  ,\enspace W>X{{\Upsilon}_{W}(Y,Z)} \right\} 
		\right\} ,\\ 					 			
	\end{aligned}
	\label{eq:SOP_MRC1}
\end{equation}
where 
\begin{equation*}
{{\Upsilon}_{W}(Y,Z)}=\frac{{\varepsilon {\lambda_{\textit{\text{SR\textsuperscript{*}}}}}}{\lambda_{\textit{\text{R\textsuperscript{*}}D}}}Y}
{{{\gamma }_{S}}{\lambda_{\textit{\text{SE}}}}\left( \varepsilon {\lambda_{\textit{\text{R\textsuperscript{*}}D}}}Y+1 \right)}
-\frac{{\varepsilon {\lambda_{\textit{\text{SR\textsuperscript{*}}}}}}{\lambda_{\textit{\text{R\textsuperscript{*}}E}}}Z}
{{\lambda_{\textit{\text{SE}}}}\left(\varepsilon {\lambda_{\textit{\text{R\textsuperscript{*}}E}}}Z+1\right)}.
\end{equation*}
\begin{figure*}[!tb]
	\begin{center}
		\begin{equation}
			\begin{aligned}
				P_{\textit{\text{out}}}^{{\text{MRC}}}\left( {{\gamma }_{S}} \right)			    = 
					&\sum\limits_{{{l}_{Y}}=0}^{{{m}_{S}}-1}{{}}{{\zeta }_{Y}} \sum\limits_{{{l}_{Z}}=0}^{{{m}_{S}}-1}{\sum\limits_{{{q}_{Z}}=0}^{{{l}_{Z}}}{{\kappa }_{Z}}}
					{{\left( {{\tilde{\gamma }}_{Z}} \right)}^{{{q}_{Z}}}} 					 {\mathcal{I}_6} \\
					& + U \sum\limits_{{{l}_{X}}=0}^{{{m}_{S}}-1}{\widetilde{\sum }_{u}^{U-1}} \sum\limits_{{{l}_{Y}}=0}^{{{m}_{S}}-1}{{}} \sum\limits_{{{l}_{Z}}=0}^{{{m}_{S}}-1}
					{{{\zeta }_{X}}{{\zeta }_{Y}} {{\zeta }_{Z}}}{{\kappa }_{u}}\Gamma \left( {{\chi }_{u}}+1 \right)  
					{{\left( \tilde{\lambda}_S \right)}^{{{\chi }_{u}}+1}}		 
					 {\mathcal{I}_7} .\\ 	
			\end{aligned} 
			\label{eq:SOP_MRC_final}
		\end{equation}	
	\end{center}
	\setlength{\arraycolsep}{1pt}
	\hrulefill \setlength{\arraycolsep}{0.0em}
\end{figure*}
\begin{figure*}[!tb]
	\begin{center}
		\begin{equation}
			\begin{split}		
				&{\begin{aligned}
						{{\Theta }_{6}}\left( u;v,\gamma ,\mu ;\alpha ,\beta  \right)={{e}^{\mu \beta }}
						&\sum\limits_{m=0}^{v}{\left( \begin{aligned}
								& v \\ 
								& m \\ 
							\end{aligned} \right)}{{\left( -\beta  \right)}^{v-m}}\sum\limits_{n=0}^{\gamma }{\left( \begin{aligned}
								& \gamma  \\ 
								& n \\ 
							\end{aligned} \right)}{{\left( \alpha -\beta  \right)}^{\gamma -n}}\\
						& \quad\quad \times \left\{ {{\Phi }_{1}}\left( \beta ;m+n-\gamma ,\mu  \right)-{{\Phi }_{1}}\left( u+\beta ;m+n-\gamma ,\mu  \right) \right\} .\\ 
				\end{aligned}} \\		
			\end{split}	
			\label{eq:Theta_6}
		\end{equation}
	\end{center}
	\setlength{\arraycolsep}{1pt}
	\hrulefill \setlength{\arraycolsep}{0.0em}
\end{figure*}
%
\begin{proposition}
Without friendly jamming, the asymptotic SOP with respect to MRC scheme at $E$ is presented in \eqref{eq:SOP_MRC_final}, where 
\begin{equation*}
	\begin{aligned}
		&{\mathcal{I}_6}={{\Theta }_{2}}\left( {{l}_{Y}},{{q}_{Z}},{{\eta }_{Y}},{{{\tilde{\eta }}}_{Z}};{{\gamma }_{S}}{{\tilde{\gamma }}_{Y}} \right) ,\\
		&\begin{aligned}
			&{\mathcal{I}_7}= \int_{0}^{\infty }{{{y}^{{{l}_{Y}}}}{{e}^{-{{\eta }_{Y}}y}}} 
			{{\Theta }_{6}}\left( {{\Upsilon }_{Y}}\left( y \right);{{l}_{Z}},{{\chi }_{u}}+1,{{\eta }_{Z}};\frac{1}{\varepsilon {{\lambda }_{\textit{\text{R\textsuperscript{*}}E}}}},
			\frac{{{\tilde{\lambda }}}_{Y} \tilde{\lambda}_S}{\varepsilon {{\lambda }_{\textit{\text{R\textsuperscript{*}}E}}}} \right)dy  \\
		\end{aligned} ,\\
	\end{aligned}
\end{equation*}	
in which defined integrals are expressed in \eqref{eq:Theta_2}, \eqref{eq:Theta_6} and 
${{\tilde{\lambda }}_{Y}}={{\eta }_{u}}
	+\frac{\varepsilon {\lambda_{\textit{\text{SR\textsuperscript{*}}}}}{\lambda_{\textit{\text{R\textsuperscript{*}}D}}}y}
	{{{\gamma }_{S}}{\lambda_{\textit{\text{SE}}}}\left( \varepsilon {\lambda_{\textit{\text{R\textsuperscript{*}}D}}}y+1 \right)}$,
	$ \tilde{\lambda}_S=\frac{{{\lambda }_{\textit{\text{SE}}}}}{{{{\tilde{\lambda }}}_{Y}}{{\lambda }_{\textit{\text{SE}}}}-{{\lambda }_{\textit{\text{SR\textsuperscript{*}}}}}} $.

\begin{proof}
	The proof is provided in Appendix D.
\end{proof}	
\end{proposition}

In case of cooperative jamming, the SOP is 
		\begin{equation}
			\begin{aligned}
			P_{\textit{\text{out}}}^{{\text{MRC,J}}}\left( {{\gamma }_{S}} \right)&={{\mathbb{E}}_{X,{\mathcal{J}}}}\left\{
			\begin{aligned}
				&\Pr \left\{ Z > {\Upsilon_Z({\mathcal{J}})} \right\} \\ & +\Pr \left\{ \begin{aligned}
				 	& Z<{\Upsilon_Z({\mathcal{J}})}, \\ 
				 	& Y<{\Upsilon_Y(Z,{\mathcal{J}})} \\ 
				 \end{aligned} \right\} 
				+\Pr \left\{ \begin{aligned}
				 	& Z<{\Upsilon_Z({\mathcal{J}})}, \\ 
				 	& Y>{\Upsilon_Y(Z,{\mathcal{J}})}, \\ 
				 	& \text{W}>{\Upsilon_W(Y,Z,{\mathcal{J}})}X \\ 
				 \end{aligned} \right\} 
			\end{aligned} \right\} ,\\ 			 				
			\end{aligned} 
			\label{eq:SOP_MRCJ2}
		\end{equation}
%
\begin{figure*}[!tb]
	\begin{center}
		\begin{equation}
			\begin{aligned}
				&P_{\textit{\text{out}}}^{{\text{MRC,J}}}
				\left( {{\gamma }_{S}} \right)=    		\sum\limits_{{{l}_{Z}}=0}^{{{m}_{S}}-1}{{}} \widehat{\sum }_{k}^{U-1}{{\zeta }_{Z}}{{\zeta }_{k}} 
				\Biggl\{ \Gamma \left( {{l}_{Z}}+1 \right)\Gamma \left( \chi _{k}^{U-1} \right){{\eta }_{Z}}^{-\left( {{l}_{Z}}+1+\chi _{k}^{U-1} \right)} \\ 
				& \quad\quad - \sum\limits_{{{l}_{Y}}=0}^{{{m}_{S}}-1} {\sum\limits_{{{q}_{Y}}=0}^{{{l}_{Y}}}{{{\kappa }_{Y}}}}{{\left( {{{\gamma }_{S}}}{\tilde{\gamma}_Y} \right)}^{{{q}_{Y}}}}
					{\mathcal{I}_8} 
					+ U \sum\limits_{{{l}_{X}}=0}^{{{m}_{S}}-1} \widetilde{\sum}_{u}^{U-1}\sum\limits_{{{l}_{Y}}=0}^{{{m}_{S}}-1}
					{{\zeta }_{X}}{{\zeta }_{Y}}{{\kappa }_{u}}
					\Gamma \left( {{\chi }_{u}}+1 \right){{\left( {\tilde{\gamma_V}} \right)}^{{{\chi }_{u}}+1}} {\mathcal{I}_9} \Biggl\} .\\
			\end{aligned} 
			\label{eq:SOP_MRCJ_final}
		\end{equation}		
	\end{center}
	\setlength{\arraycolsep}{1pt}
	\hrulefill \setlength{\arraycolsep}{0.0em}
\end{figure*}

where
\begin{equation*}
	\begin{aligned}		 
	& {\Upsilon_Z({\mathcal{J}})}=\frac{1+{P_J}{\lambda_{\textit{\text{R\textsuperscript{*}}E}}}{\mathcal{J}}}
	{\varepsilon {\lambda_{\textit{\text{R\textsuperscript{*}}E}}}\left( {{\gamma }_{S}}-1 \right)} , \enspace {\Upsilon_Y(Z,{\mathcal{J}})}=\frac{{{\gamma }_{S}}{{V}(Z,{\mathcal{J}})}}
	{{\lambda_{\textit{\text{R\textsuperscript{*}}D}}}\left( 1-\varepsilon {{\gamma }_{S}}{{V}(Z,{\mathcal{J}})} \right)} ,\\
		&{\Upsilon_W(Y,Z,{\mathcal{J}})}=\varepsilon {\lambda_{\textit{\text{SR\textsuperscript{*}}}}}\left( \frac{{\lambda_{\textit{\text{R\textsuperscript{*}}D}}}Y}
		{{{\gamma }_{S}}\left(1+\varepsilon {\lambda_{\textit{\text{R\textsuperscript{*}}D}}}Y \right)}-{{V}(Z,{\mathcal{J}})} \right) ,\\ 
		&{{V}(Z,{\mathcal{J}})}=\frac{{\lambda_{\textit{\text{R\textsuperscript{*}}E}}}Z}
		{\varepsilon {\lambda_{\textit{\text{R\textsuperscript{*}}E}}}Z+1+{P_J}{\lambda_{\textit{\text{R\textsuperscript{*}}E}}}{{\mathcal{J}}}} .\\ 				
	\end{aligned}
\end{equation*}
\begin{proposition}
With friendly jamming, the asymptotic SOP with respect to MRC scheme at $E$ is presented in \eqref{eq:SOP_MRCJ_final}, where 
\begin{equation*}
	\begin{aligned}
		&\begin{aligned}
			{\mathcal{I}_8}= &\int_{0}^{\infty }{{{t}^{\chi _{k}^{U-1}-1}}{{e}^{-{{\eta }_{Z}}t}}} 
			{{\Theta }_{4}}\left( {{\Upsilon }_{Z}}\left( t \right);{{l}_{Z}},{{q}_{Y}},{{\eta }_{Z}},{{{\tilde{\eta }}}_{Y}};{{\Upsilon }_{Z}}\left( t \right) \right)dt \\
		\end{aligned} ,\\
		&\begin{aligned}
			{\mathcal{I}_9}= 
			& \int_{0}^{\infty }{{{t}^{\chi _{k}^{U-1}-1}}{{e}^{-{{\eta }_{Z}}t}}}\int_{0}^{{{\Upsilon }_{Z}}\left( t \right)}{{{z}^{{{l}_{Z}}}}{{e}^{-{{\eta }_{Z}}z}}} {{\Theta }_{5}} \left( {{\Upsilon }_{Y}}\left( z,t \right);{{l}_{Y}},{{\chi }_{u}}+1,{{\eta }_{Y}};\frac{1}{\varepsilon {\lambda_{\textit{\text{R\textsuperscript{*}}D}}}},\frac{{\tilde{\gamma_V}}{{\lambda }_{\textit{\text{SE}}}}\tilde{V}\left( z,t \right)}{\varepsilon {{\lambda }_{\textit{\text{R\textsuperscript{*}}D}}}} \right) dzdt ,\\
		\end{aligned} \\
	\end{aligned}
\end{equation*}	
in which defined integrals are expressed in \eqref{eq:Theta_4}, \eqref{eq:Theta_5}
and $\tilde{V}\left( z,t \right)={{\eta }_{u}}-\varepsilon\frac{{{\lambda }_{\textit{\text{SR\textsuperscript{*}}}}}}{\lambda_{\textit{\text{SE}}}}V\left( z,t \right)$,
${\tilde{\gamma_V}}=\frac {{{\gamma }_{S}}}{{{\lambda }_{\textit{\text{SR\textsuperscript{*}}}}}+{{\gamma }_{S}}{{\lambda }_{\textit{\text{SE}}}}\tilde{V}\left( z,t \right)}$.
\end{proposition}
\begin{proof}
	The proof is provided in Appendix E.
\end{proof}
%
\begin{figure}[t!]
	\centerline{\includegraphics[width=0.75\linewidth]{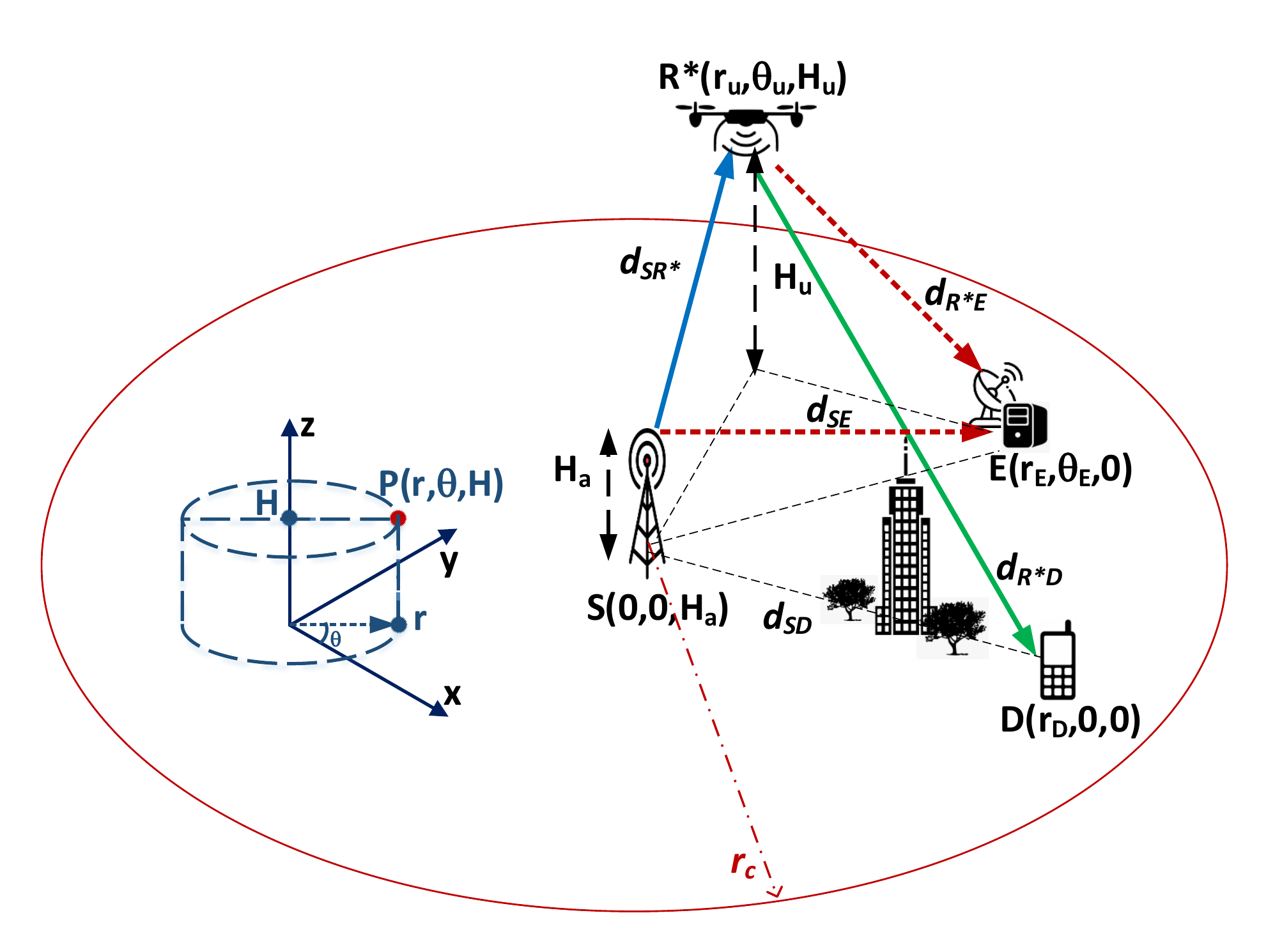}}
	\caption{The illustration of distance model.}
	\label{fig:distance}
\end{figure}
\section{General Case: Unknown Eavesdropper's CSI}
\label{Sec:Unknown}
In this section, we consider the case that an eavesdropper is uniformly distributed on ground inside a circular disc of radius $r_c$ around the source node. 
\subsection{Distance Model}
We use the polar coordinate system to facilitate the analysis with the coordinate origin at the source at $\mathbf{p_S}=(0,0,H_a)$, where $H_a$ is the height of the antenna tower from the ground as illustrated in {Fig.}~\ref{fig:distance}. {UAVs are} flying at $\mathbf{p_u}=(r_u,\theta_u,H_u)$ to serve $D$ which is located at $\mathbf{p_D}=(r_{D},0,0)$. 
The corresponding distance from $S$ to $R_u$ and $R_u$ to $D$ are $d_{\textit{\text{Su}}}=\sqrt{{r_u}^2+\left({H_u}-{H_a}\right)^2}$ and $d_{\textit{\text{uD}}}=\sqrt{{r_u}^2+{r_{D}}^2-2r_{D}{r_u}\cos{\theta_u}+{H_u}^2}$, respectively. 

Since the eavesdropper's location is unknown in practical applications, the stochastic geometry is used to describe the eavesdroppers' location in a specific environment \cite{c._liu_artificial-noise-aided_2016}. We consider the case $E$ at $\mathbf{p_E}=(r_{E},\theta_E,0)$ is uniformly distributed inside a circular disc of radius $r_c$ that is centered at the origin $S$, ${r_E \le r_c}, \enspace 0 \le \theta_E \le 2\pi$. The distance from $S$ and $R_u$ to $E$ are $d_{\textit{\text{SE}}}=\sqrt{{r_E}^2+{{H_a}}^2}$ and $d_{\textit{\text{uE}}}=\sqrt{{r_u}^2+{r_E}^2-2r_ur_E\cos{\left(\theta_u-\theta_E\right)}+{H_u}^2}$, respectively.
The distribution of $E$ is modeled by the binomial point process (BPP) $\Phi_E$ with the corresponding PDF as \cite{7882710}
\begin{equation}
	\begin{aligned}
		{{f}_\mathbf{p_E}}\left( r_E,\theta_E \right)=\frac{1}{\pi{{r_c}^{2}}}, \enspace {r_E \le r_c}, \enspace 0 \le \theta_E \le 2\pi.
	\end{aligned}
\end{equation}
\subsection{Lower-Bound SOP}
The CSI between $S$ and $E$ depends on $E$'s location via the free-space path losses, i.e., ${{\lambda }_\textit{SE}}\propto {{d}_\textit{SE}}\left( {{r}_{E}} \right)$, ${{\lambda }_{\textit{\text{R\textsuperscript{*}}E}}}\propto {{d}_{\textit{\text{R\textsuperscript{*}}E}}}\left( {{r}_{E}},\theta_E \right)$.
Therefore, the system SOP in the presence of a randomly distributed eavesdropper can be written as  
\begin{equation}
	\begin{aligned}
		{{{\bar{P}}}_{\textit{\text{out}}}}\left( {{\gamma }_{S}} \right)&={{\mathbb{E}}_{\mathbf{p_E}}}\left\{ {{{P}}_{\textit{\text{out}}}}\left( {{\gamma }_{S}},{{d}_\textit{SE}},{{d}_{\textit{\text{R\textsuperscript{*}}E}}} \right) \right\} ,\\
	\end{aligned}
	\label{eq:SOP,UE}
\end{equation}
where ${{{P}}_{\textit{\text{out}}}}\left( {{\gamma }_{S}},{{d}_\textit{SE}},{{d}_{\textit{\text{R\textsuperscript{*}}E}}} \right)$ is the expression of SOP w.r.t. to {the} fixed location of $E$, which is derived in Section~\ref{Sec:Secrecy} for different cases of combining schemes of SC or MRC at $E$. 
%
Although we cannot find the closed-form expression in this case, one can rely on numerical tools or apply the Jensen’s inequality on \eqref{eq:SOP,UE} to effectively evaluate the lower bound SOP \cite{1512427}.

The lower bound SOP is obtained by replacing ${d_{\textit{\text{S}E}}\left(r_E \right)}$ and ${d_{\textit{\text{R\textsuperscript{*}}E}}\left(r_E,\theta_E \right)}$ with $\mathcal{R}_{\textit{\text{S}E}}={{\mathbb{E}}_{\mathbf{p_E}}} \left\{ {d_{\textit{\text{S}E}}\left(r_E \right)} \right\}$ and $\mathcal{R}_{\textit{\text{R\textsuperscript{*}}E}}={{\mathbb{E}}_{\mathbf{p_E}}} \left\{ {d_{\textit{\text{R\textsuperscript{*}}E}}\left(r_E,\theta_E \right)} \right\}$ in \eqref{eq:SOP_SC_final}, \eqref{eq:SOP_SCJ_final}, \eqref{eq:SOP_MRC_final} and \eqref{eq:SOP_MRCJ_final}. $\mathcal{R}_{\textit{\text{S}E}}$ and $\mathcal{R}_{\textit{\text{R\textsuperscript{*}}E}}$ are calculated as
\begin{equation}
	\mathcal{R}_{\textit{\text{S}E}}=\int_{0}^{2\pi }{\int_{0}^{r_c}}{{{d}_\textit{SE}}\left( {{r}} \right){{f}_\mathbf{p_E}}\left( r,\theta \right)}rdrd\theta, \enspace
%
		\mathcal{R}_{\textit{\text{R\textsuperscript{*}}E}}=\int_{0}^{2\pi }{\int_{0}^{r_c}}{{{d}_{\textit{\text{R\textsuperscript{*}}E}}}\left( {{r}},\theta \right)}
		{{f}_\mathbf{p_E}\left( r,\theta \right)}
		rdrd\theta.\\
	\label{eq:E_dRE}
\end{equation}
The proof for the case of selection combining without and with jamming is provided in Appendix F.



%
\section{A {Case Study}: Optimal Area for Locating UAVs}
\label{Sec:Case}
In this section, we apply the above SOP analysis to optimize the corridor, where a swarm of UAVs should locate.
The optimal corridor is defined as $\mathbf{p_u}=\left(r_u,\theta_u,H_u\right)$,
to locate {a swarm of UAVs} over which the system SOP {at a certain value of the target secrecy SNR} is minimized.
{We consider that UAVs can change their altitudes within a constraint of $H_u \in [H_{\min},H_{\max}]$ for safety considerations such as terrain or building and airplane avoidance \cite{g._zhang_securing_2019}.}
%
{From the distance model, the CSI between UAVs and other nodes depends on UAVs' locations via the free-space path losses, i.e., ${{\lambda }_\textit{Su}}\propto {{d}_\textit{Su}}\left( \mathbf{p_u} \right)$, ${{\lambda }_{\textit{\text{uD}}}}\propto {{d}_{\textit{\text{uD}}}}\left( \mathbf{p_u} \right)$ and ${{\lambda }_{\textit{\text{uE}}}}\propto {{d}_{\textit{\text{uE}}}}\left( \mathbf{p_u},\mathbf{p_E} \right)$.}
%
{The lower bound SOP in {equation \eqref{eq:SOP,UE}} can be rewritten as related to the location of swarm}
	\begin{equation}
	\begin{aligned}
	{{{\bar{P}}}_{\textit{\text{out}}}}\left( {{\gamma }_{S}},\mathbf{p_u} \right)
	&={{\mathbb{E}}_{\mathbf{p_E}}}\left\{ {{{P}}_{\textit{\text{out}}}}\left( {{\gamma }_{S}},\mathbf{p_u},\mathbf{p_E} \right) \right\}.\\
	\end{aligned}
	\label{eq:SOP,pu}
	\end{equation}
Therefore, the problem to find the optimal 
location can be expressed as  
{\begin{equation}
	\begin{aligned}
		& \underset{\mathbf{p_u}}{\text{min}} \quad {{\bar{P}}_\textit{\text{out}}^{{\text{MRC,J}}}}\left( {{\gamma }_{S}},\mathbf{p_u} \right) 
		\\ 
		& \begin{aligned}
			& s.t. \quad {H_{\min}}\le{H_u}\le{H_{\max}} ,\\ 
			& \quad\quad {0} \le {r_u} \le {r_c} , \enspace {0} \le {\theta_u} \le {2\pi} .\\
		\end{aligned} \\ 
	\end{aligned}
	\label{eq:opt}
\end{equation}}

Due to the complexity of the expressions of { ${{\bar{P}}_\textit{\text{out}}^{{\text{MRC,J}}}}\left( {{\gamma }_{S}},\mathbf{p_u} \right)$ in \eqref{eq:SOP,pu} with the fact that the position of $E$ is unknown}, the solution for the optimal location is intractable. Instead, we can obtain the optimal location with the aid of exhaustive search programming. 
To decrease the complexity of the search problem, we begin with the case that $E$ attempts to eavesdrop information from both the source and the relay UAV, $E$ then can be considered randomly located on either side of the $S$-to-$R$ connected line. The capacity of the eavesdropping channel is independent of the $R$-to-$D$ distance. Therefore, the swarm of UAVs can be located along the straight line between $S$ and $D$, $\theta_u=0$, to increase the capacity of the legitimate link, in order to minimize the SOP. Therefore, our 3D optimization problem can be reduced to 2D optimization one. 
The flying altitude and x-coordinate are varied to find the optimal locations of {UAVs} to achieve the minimum SOP.

	\begin{figure*}[!tb]
		\begin{center}
			\minipage{0.47\textwidth}
			\includegraphics[width=\linewidth]{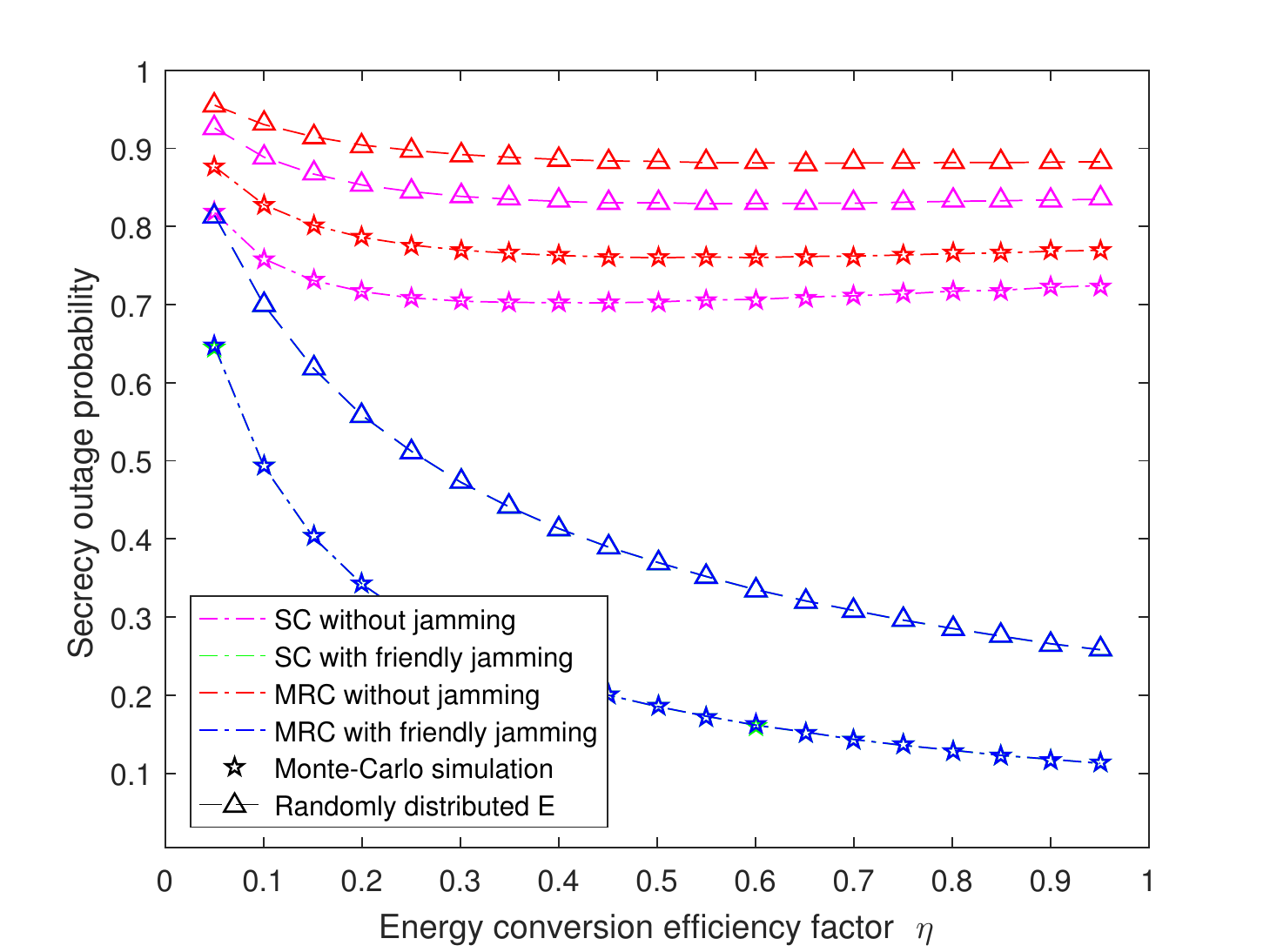}
			\subcaption{}
			\endminipage
			\hfill
			\minipage{0.47\textwidth}
			\includegraphics[width=\linewidth]{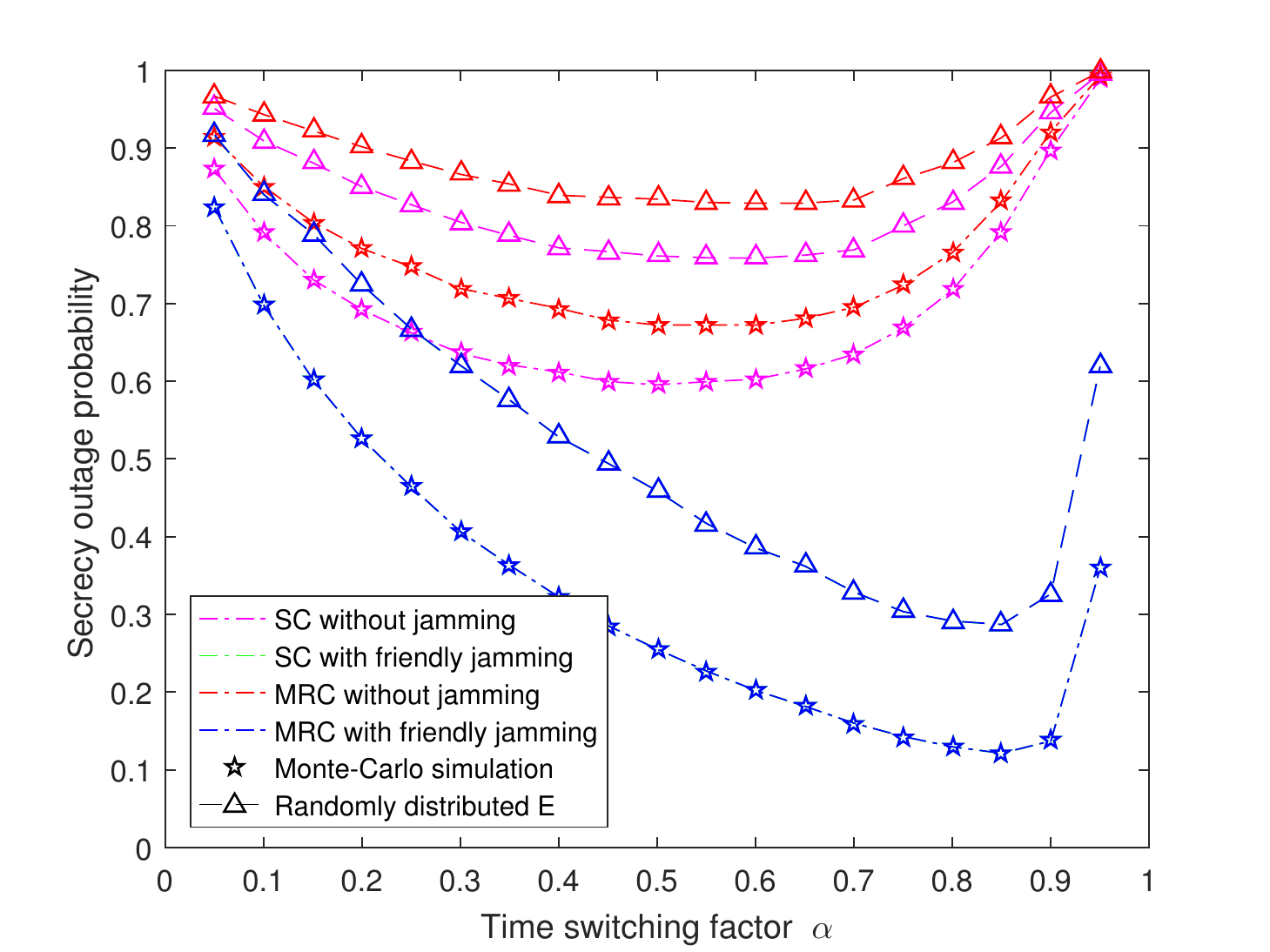}
			\subcaption{}
			\endminipage
			\hfill \\
    		\minipage{0.47\textwidth}%
			\includegraphics[width=\linewidth]{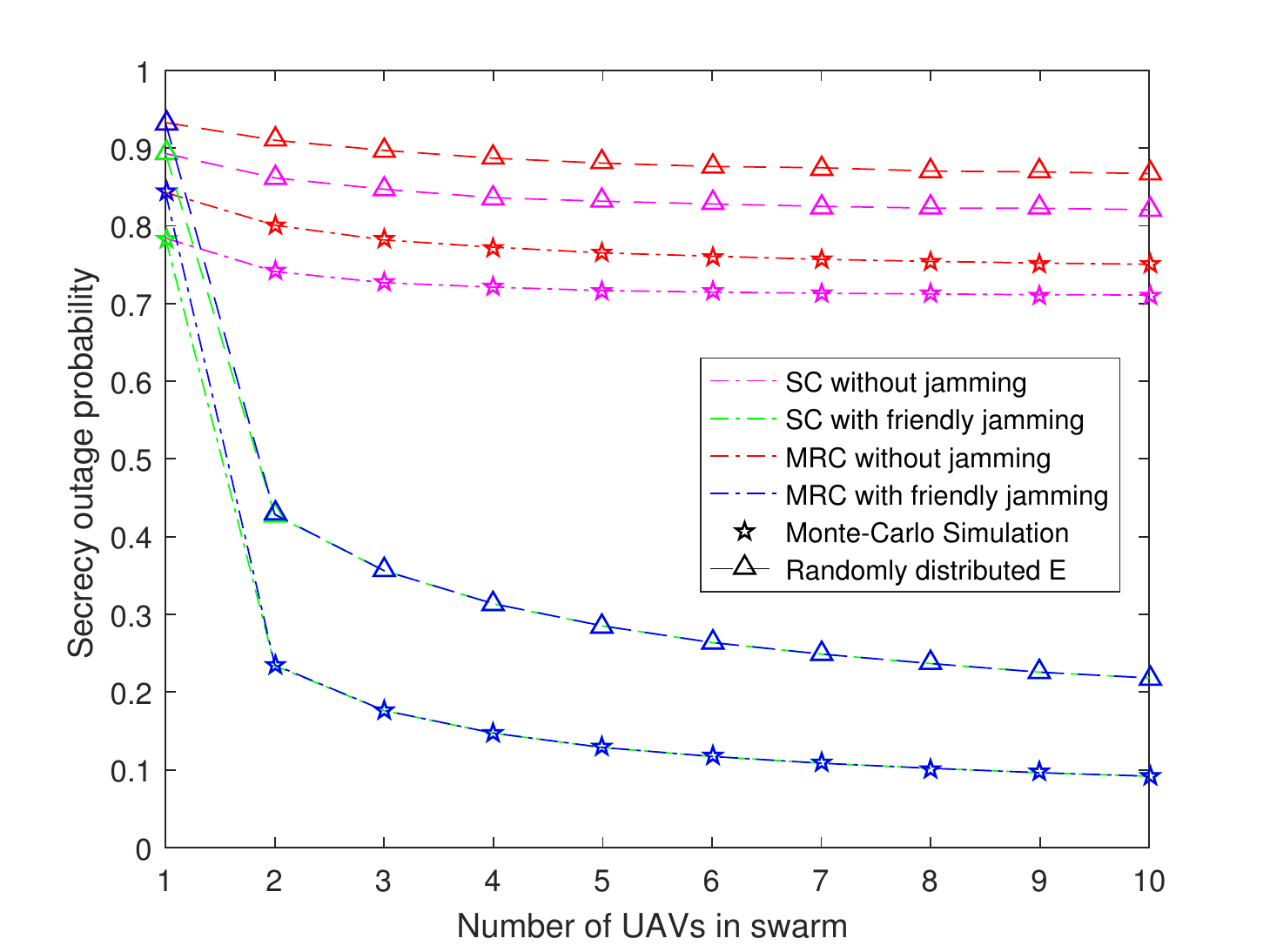}
			\subcaption{}
			\endminipage
			\caption{Secrecy outage probability vs. (a) EH coefficient; (b) time-switching; (c) number of UAVs in swarm.}
			\label{fig:SOP}	
		\end{center}
		\setlength{\arraycolsep}{1pt}
		\hrulefill \setlength{\arraycolsep}{0.0em}
	\end{figure*}
\section{Performance Evaluation}
\label{Sec:Results}
In this section, Monte Carlo simulations are conducted to validate the theoretical expressions in \eqref{eq:SOP_SC_final}, \eqref{eq:SOP_SCJ_final}, \eqref{eq:SOP_MRC_final}, \eqref{eq:SOP_MRCJ_final}, {and \eqref{eq:SOP,UE}} as well as to obtain insights into the secrecy performance of our system. 
We first set the target secrecy capacity at ${C}_{th}=0.1$ bps/Hz. 
We assume that the jamming UAVs use all their harvested energy to effectively jam the eavesdropper, i.e., $\delta = 1$. 
For the purpose of illustration, all the coordinate systems are presented in Cartesian and in meters. $S$ and $D$ are fixed at $\mathbf{p_S}=(300,300,25)$ and $\mathbf{p_D}=(600,300,0)$, respectively. The channel between UAVs and terrestrial BSs are under shadowed-Rician fading 
of $(m_S,b,\Omega)=(5, 0.251, 0.279)$ for average shadowing. The free-space path loss  
refers to a reference distance at $d_0=100$ meters. 
\subsection{Secrecy Performance Analysis}
To illustrate the secrecy performance of the system, the SOP is investigated for two techniques of SC and MRC at $E$ in a high SNR of $\psi={P_S}/{{\sigma}^2 }=40 \text{dB}$. In these simulations, {UAVs} hover at $H_u=60$ meters, in particular at $\mathbf{p_{u}}=(350,300,60)$. Whereas $E$ is assumed to be {known and located} at $\mathbf{p_E}=(600,400,0)$ or is randomly located around $S$ inside a circular disk of $r_c=300$ meters. 

{Figs}.~\ref{fig:SOP}(a) and (b) show the impact of the energy conversion efficiency factor $\eta$ and the time-switching factor $\alpha$ on the system SOP, when $U=5$ UAVs in the swarm. 
Fig.~\ref{fig:SOP}(c) depicts the SOP vs. the number of UAVs when $\alpha=0.8$ and $\eta=0.8$. 
These figures validate the theoretical derivation in \eqref{eq:SOP_SCJ_final} and \eqref{eq:SOP_MRCJ_final} as the coincidence between the asymptotic analysis and the Monte Carlo simulation. The effectiveness of our proposed model is presented by the more deviation between the two cases of without and with friendly jamming.
Without friendly jamming, the SOP is not impacted much by increasing $\eta$ or $U$ in case of known $E$, {i.e.,} {the SOP is around $0.4$ when $\eta \ge 0.25$ or $U \ge 3$}.   
{In our proposed system, the relay selection 
	{can provide} 
	the better link between $S$ and UAVs when the number of UAVs increases, while the more harvested energy for amplifying and forwarding the information to $D$ provides the better signal at $D$. Unfortunately, those also benefit the eavesdropping link from UAVs to $E$. Therefore, at a certain value of $\eta$ and $U$,} the signals received at $E$ from UAVs is as good as at $D$, {resulting in a less change in the SOP.}  
With friendly jamming, the higher $\eta$ or $\alpha$, or the more number of UAVs in the swarm provides a more harvested energy to jam $E$. {Therefore, the channels between UAVs and $E$ are severely degraded.}
This results in {a significant} decrease in the SOP. 

{From Fig.~\ref{fig:SOP}(b), increasing $\alpha$ at first gives more energy to relay the information and jam the eavesdropper, as such improves the system SOP. It is recalled that $\alpha$ is the fraction of time for energy harvesting over the total transmission time. Consequently, the more the time for energy harvesting, the less the time for signal transmission. Therefore, the time switching factor can be optimized to achieve the minimum SOP, e.g., $\alpha=0.5$ without friendly jamming or $\alpha=0.85$ with friendly jamming.}
{Furthermore, the eavesdropper is the most effective in overhearing the legitimate transmission using the MRC scheme. Without friendly jamming, the SOP in the case of the MRC scheme at $E$ is greater than that of the SC scheme at $E$ as expected. Using friendly jamming to degrade the channels between UAVs and $E$, results in the coincidence of two schemes. This implies signals from UAVs are not significant as compared to the signals from $S$.}%

\begin{figure}[!tb]
  \begin{center}
  \begin{minipage}{0.45\textwidth}
    {\includegraphics[width=\linewidth]{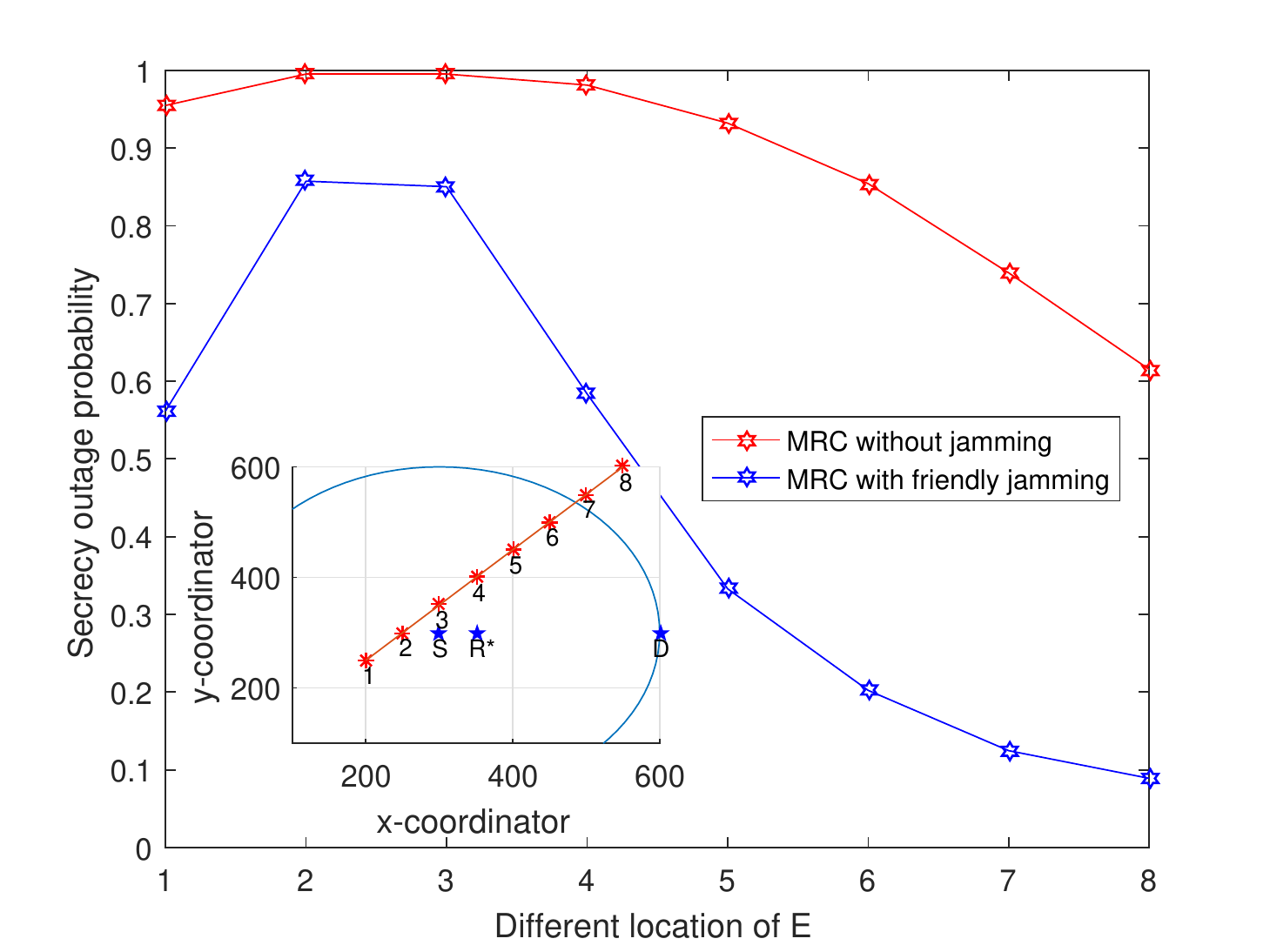}}
	\caption{SOP versus $E$'s location.}
	\label{fig:SOP_disE}
  \end{minipage}
  \hfill
  \begin{minipage}{0.45\textwidth}
    {\includegraphics[width=\linewidth]{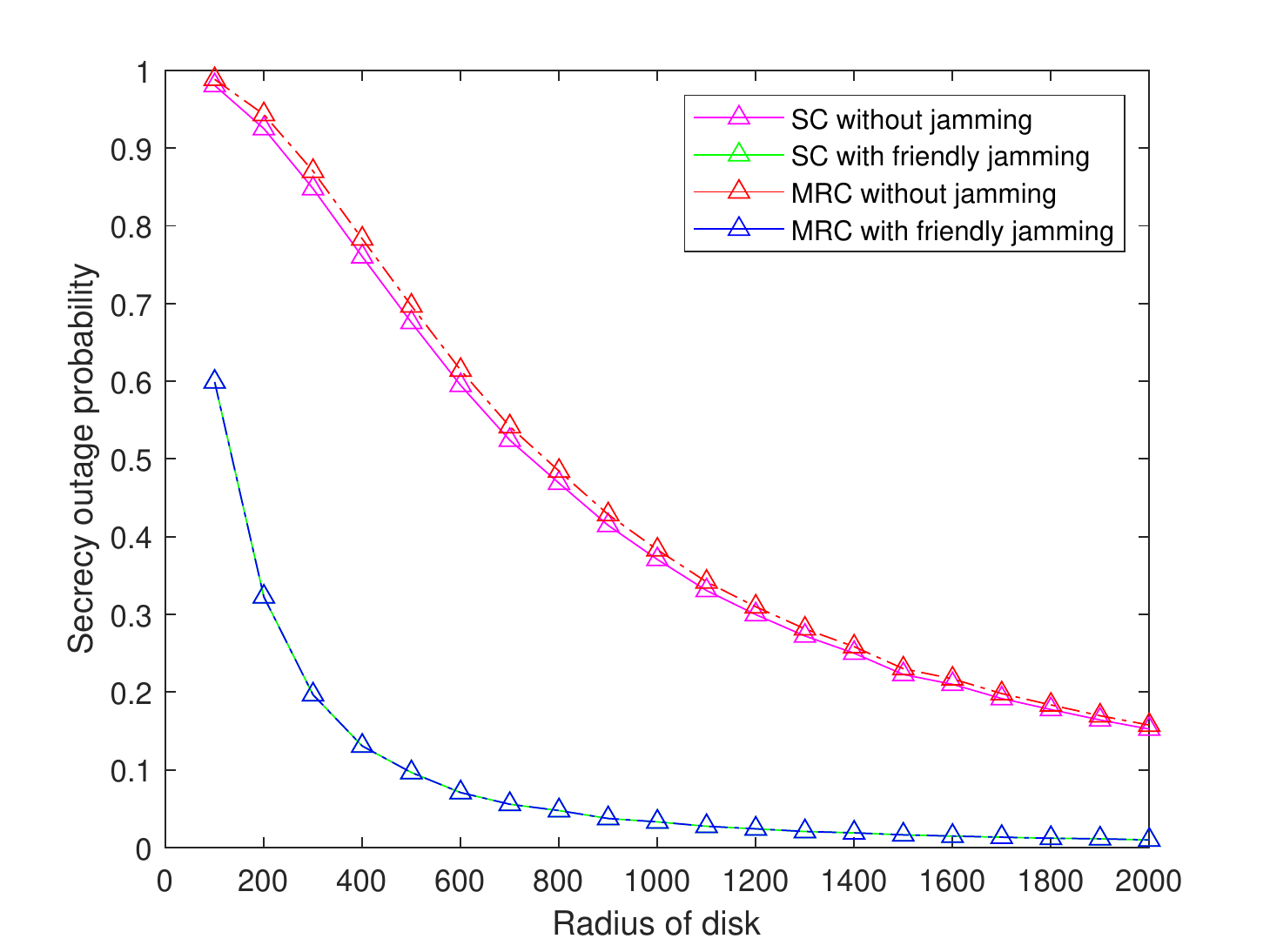}}
	\caption{SOP versus the radius $r_c$.}
	\label{fig:SOP_Rmax}
  \end{minipage}
  \end{center}
\end{figure}
We extend the study to the case of a randomly distributed $E$.  
%
Fig.~\ref{fig:SOP_disE} shows the significant impact of $E$'s location on the SOP. 
The particular locations of all nodes are illustrated on the subplot. $E$ is assumed to change its location from $(550,600,0)$ to $S$, numbered from 1 to 8. The SOP increases as $E$ gets closer to $S$ (i.e., $E$'s location numbered 2 and 3) and decreases when $E$ goes far away from $S$. 
{Locating around $S$ helps $E$ improve its capacity by boosting the link from $S$ to $E$.
	In general case of unknown $E$, the system SOP is calculated as the average SOP over the region in which $E$ is randomly distributed.}
As expected, from Figs.~\ref{fig:SOP}(a), (b), (c), and {Fig.~\ref{fig:SOP_Rmax},} the system SOP for this case is higher than the case when $E$ is fixed far away from $S$. These figures also {show} that using UAVs to jam $E$, {which is randomly distributed,} provides the lower SOP as compared to the case without jamming.%

In above figures, $E$'s geometry environment is assumed to be a circular disk around $S$ with radius $r_c=d_\textit{SD}$. However, as $E$ is randomly distributed, we cannot exactly know $r_c$. Therefore, the impact of $E$'s distributed environment on the SOP is illustrated in {Fig.~\ref{fig:SOP_Rmax}}. The SOP is high when $E$ is distributed around $S$, and significantly decreases when the considered circular disk is larger. For example, when $r_c$ is greater than $500$ meters, SOP is less than 0.1. 
{Despite the fact that $E$'s distributed environment is unknown, our proposed model still proves its efficiency when showing the significant deviation between the two cases of without and with friendly jamming.}
%
\begin{figure*}[!tb]
	\begin{center}
		\minipage{0.49\textwidth}%
		\includegraphics[width=\linewidth]{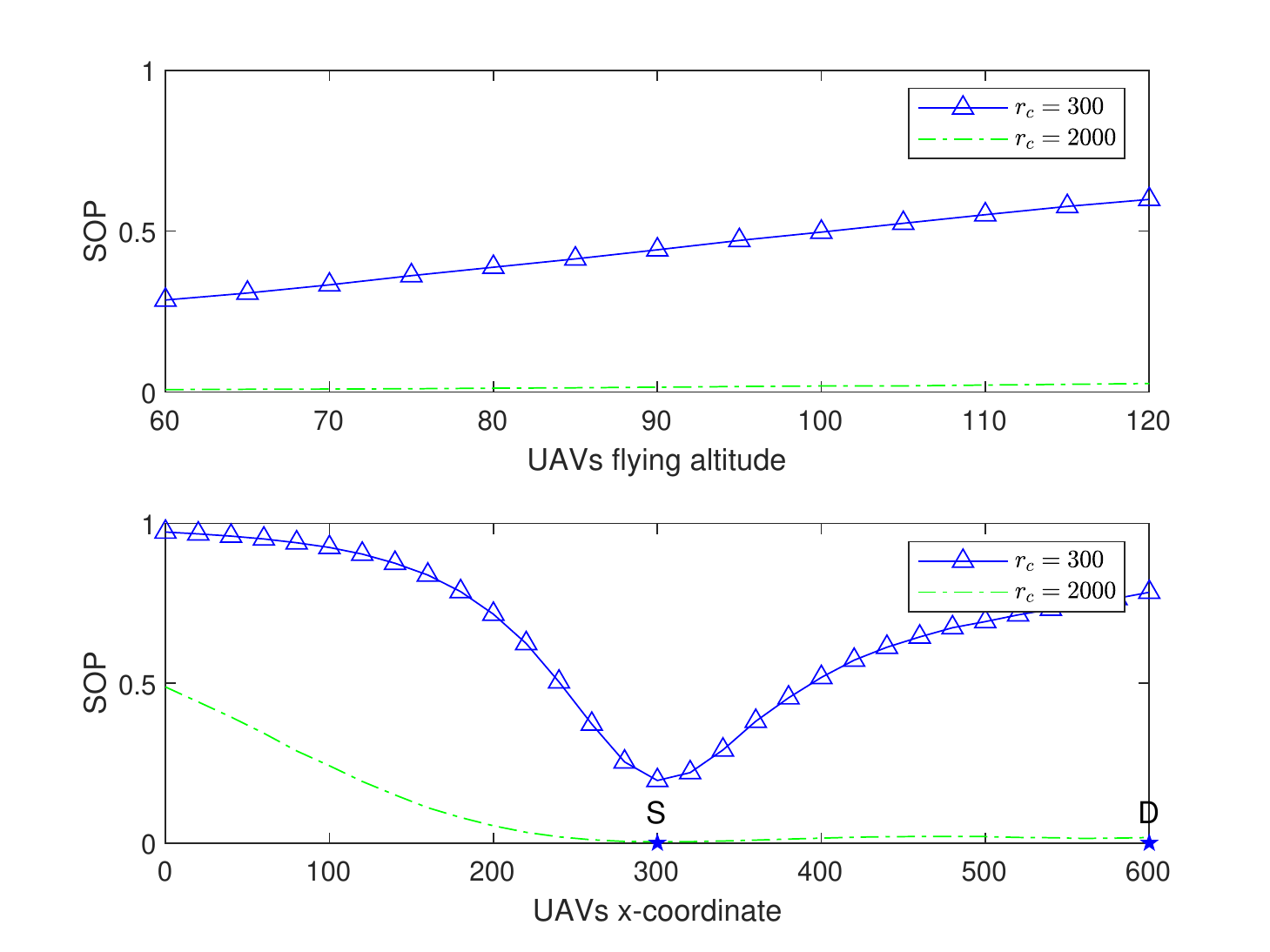}
		\subcaption{SOP versus $z_u$, $r_u$.}
		\endminipage
		\hfill				
		\minipage{0.49\textwidth}
		\includegraphics[width=\linewidth]{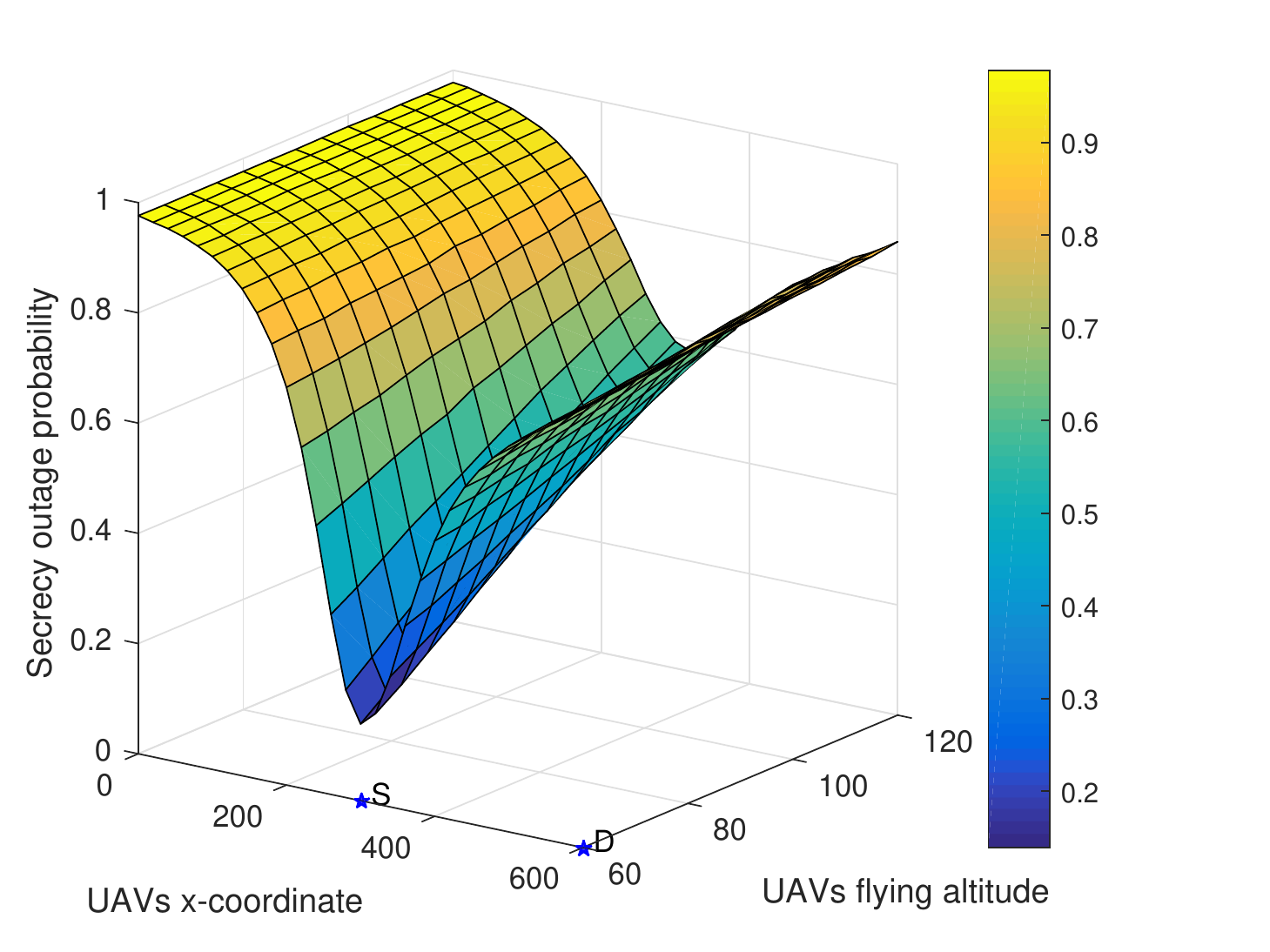}
		\subcaption{SOP versus $\mathbf{p_u}$.}
		\endminipage
		\caption{Secrecy outage probability vs. (a) the flying altitude and the x-coordinator of UAVs; (b) the 2D location of UAVs.}
		\label{fig:SOP_Loc}	
	\end{center}
	\setlength{\arraycolsep}{1pt}
	\hrulefill \setlength{\arraycolsep}{0.0em}
\end{figure*}
%
\subsection{Finding Optimal Corridor for Locating UAVs }

The impact of UAVs' locations on the system SOP in the presence of a randomly distributed $E$ is shown in Fig.~\ref{fig:SOP_Loc}(a). The flying altitude and x-coordinate of UAVs are independently investigated while other coordinates remain the same as above. 
Since higher flight altitudes lead to a more severe path loss for ground-to-air and air-to-ground communications, the SOP increases along with the higher flight altitude.
Due to the free-space path loss model, when the distance between $S$ and $R$ increases, UAVs receive lower SNR signals from $S$, and harvest less RF energy to relay signals to $D$ as well as to jam $E$. 
{In addition, from Fig.~\ref{fig:SOP_disE}, the system SOP is the worst when $E$ is located around $S$. To effectively jam $E$, jamming UAVs should hover above $S$.}
Therefore, the SOP increases when UAVs fly at a high altitude and far from $S$.  

A {case study} of finding the optimal corridor for locating UAVs is presented in Fig.~\ref{fig:SOP_Loc}(b). The SOP is now varied with both the flying altitude and the x-coordinate in a 3D plot. 
Fig.~\ref{fig:SOP_Loc}(b) shows that the SOP reaches its minimum value when the swarm of UAVs {flies} at the lowest allowed altitude and right above $S$ at $(300,300,60)$. 
{From these illustrations, in order to guarantee a level of SOP, UAVs should hover above $S$ and at the lowest accepted altitude.}
\section{Conclusions}
\label{Sec:Conclusion}
In this paper, we proposed a cooperative friendly jamming in swarm UAV-assisted communications with wireless energy harvesting. The theoretical SOPs were derived for two popular combining schemes at the eavesdropper and verified by Monte Carlo simulations. Using the SOP, we obtained engineering insights to optimize the energy harvesting time, the number of UAVs in the swarm to achieve the required secrecy level. Furthermore, we investigated a {case study} that uses the derived SOP to find the optimal corridor to locate the swarm {of hovering UAVs} so as to minimize the secrecy outage probability in the presence of an eavesdropper.
\appendices
\section{Proof of Lemma 2}
\label{app:Lem2}
We adopt the Moment Generating Function approach to derive the PDF of ${{\mathcal{J}}}=\sum\limits_{j=1}^{U-1}{{{\left| {{h}_{jE}} \right|}^{2}}}$ as follows. Since the aerial-based links between UAV-aided jammers and $E$ are under shadowed-Rician fading, the PDF of each jamming channel gain, ${{\left| {{h}_{jE}} \right|}^{2}}$, is written as
\begin{equation}
	{{f}_{{{\left| {{h}_{jE}} \right|}^{2}}}}\left( x \right)={{f}_{Zj}}\left( x \right)=\sum\limits_{{{k}_{j}}=0}^{{{m}_{S}}-1}{{{\zeta }_{Z}}}{{x}^{{{k}_{j}}}}{{e}^{-{{\eta }_{Z}}x}},
\end{equation}
Using \cite[eq. (3.351,3)]{zwillinger_3-4_2015}, the Laplace transform of ${{f}_{Z_j}}\left( x \right)$ is 
\begin{equation}
	\begin{split}
		{{\hat{P}}_{Z_j}}\left( s \right)
		&=\int_{0}^{\infty }{{{f}_{Z_j}}\left( x \right)}{{e}^{-sx}}dx
		=\sum\limits_{{{k}_{j}}=0}^{{{m}_{S}}-1}{{{\zeta }_{Z}}}{{k}_{j}}!{{\left( {{\eta }_{Z}}+s \right)}^{-\left( {{k}_{j}}+1 \right)}}.
	\end{split}
\end{equation}
Therefore, the Laplace transform of ${{f}_{{\mathcal{J}}}}\left( x \right)$ can be rewritten as 
\begin{equation}
	\begin{aligned}
		{{{\hat{P}}}_{{\mathcal{J}}}}\left( s \right)
		&=\prod\limits_{j=1}^{U-1}{{{{\hat{P}}}_{Z_j}}\left( s \right)} =\sum\limits_{{{k}_{1}}=0}^{{{m}_{S}}-1}{\ldots \sum\limits_{{{k}_{U-1}}=0}^{{{m}_{S}}-1}{\left[ \prod\limits_{j=1}^{{U-1}}{\left( {{\zeta }_{Z}}{{k}_{j}}! \right)} \right]}}{{\left( {{\eta }_{Z}}+s \right)}^{-\sum\limits_{j=1}^{U-1}{\left( {{k}_{j}}+1 \right)}}}. \\ 
	\end{aligned}
\end{equation}
Finally, applying the inverse Laplace transform of \eqref{eq:L-inverse}, we obtain the PDF of ${{\mathcal{J}}}$ as in \eqref{eq:PDF_Zsum} to complete the proof of Lemma 3.
\begin{equation}
	{{\mathcal{L}}^{-1}}\left\{ {{\left( \frac{1}{s+{1}/{\lambda }\;} \right)}^{n}} \right\}=\frac{1}{\left( n-1 \right)!}{{t}^{n-1}}{{e}^{-\frac{t}{\lambda }}}.
	\label{eq:L-inverse}
\end{equation}

\begin{figure*}[!tb]
	\begin{center}
		\begin{equation}
			\begin{aligned}
P_{\textit{\text{out}}}^{{\text{SC}}}\left( {{\gamma }_{S}} \right)
& =1-U\sum\limits_{{{l}_{X}}=0}^{{{m}_{S}}-1}{\widetilde{\sum }_{u}^{U-1}}\sum\limits_{{{l}_{Y}}=0}^{{{m}_{S}}-1}{{{\zeta }_{X}}{{\zeta }_{Y}}}{{\kappa }_{u}}\Gamma \left( {{\chi }_{u}}+1 \right) \\ 
& \times \left\{ \begin{aligned}
	& {{\eta }_{u}}^{-\left( {{\chi }_{u}}+1 \right)}\Gamma \left( {{l}_{Y}}+1 \right){{\eta }_{Y}}^{-\left( {{l}_{Y}}+1 \right)}
	-\underbrace{\int_{0}^{\infty }{{{y}^{{{l}_{Y}}}}{{e}^{-{{\eta }_{Y}}y}}}{{\left( {{\eta }_{u}}+{{\Upsilon }_{W}}\left( y \right) \right)}^{-\left( {{\chi }_{u}}+1 \right)}}dy}_{{{J}_{1}}} \\ 
	& -\sum\limits_{{{l}_{Z}}=0}^{{{m}_{S}}-1}{\sum\limits_{{{q}_{Z}}=0}^{{{l}_{Z}}}{{{\kappa }_{Z}}}}{{\eta }_{u}}^{-\left( {{\chi }_{u}}+1 \right)}
	\underbrace{\int_{0}^{\infty }{{{y}^{{{l}_{Y}}}}{{e}^{-{{\eta }_{Y}}y}}{{\left\{ {{\Upsilon }_{Z}}\left( y \right) \right\}}^{{{q}_{Z}}}}{{e}^{-{{\eta }_{Z}}{{\Upsilon }_{Z}}\left( y \right)}}}dy}_{{{J}_{2}}} \\ 
	& -\sum\limits_{{{l}_{Z}}=0}^{{{m}_{S}}-1}{\sum\limits_{{{q}_{Z}}=0}^{{{l}_{Z}}}{{{\kappa }_{Z}}}}
	\underbrace{\int_{0}^{\infty }{{{y}^{{{l}_{Y}}}}{{e}^{-{{\eta }_{Y}}y}}}{{\left( {{\eta }_{u}}+{{\Upsilon }_{W}}\left( y \right) \right)}^{-\left( {{\chi }_{u}}+1 \right)}}{{\left\{ {{\Upsilon }_{Z}}\left( y \right) \right\}}^{{{q}_{Z}}}}{{e}^{-{{\eta }_{Z}}{{\Upsilon }_{Z}}\left( y \right)}}dy}_{{{J}_{3}}} \\ 
\end{aligned} \right\} \\ 
			\end{aligned} 
			\label{eq:SOP_SC3}
		\end{equation}
	\end{center}
	\setlength{\arraycolsep}{1pt}
	\hrulefill \setlength{\arraycolsep}{0.0em}
\end{figure*}
\section{Proof of Proposition 1}
\label{app:Pro1}
Using the CDF and PDF of all r.v. to rewrite \eqref{eq:SOP_SC1}, after some simplifications and employing \cite[eq. (3.471,9)]{zwillinger_3-4_2015}, %
we have \eqref{eq:SOP_SC3}. 
The first integral is rewritten as
\begin{equation}
	\begin{aligned}
{{J}_{1}}
& ={{\left( {{\tilde{\gamma }}_{SE}} \right)}^{ {{\chi }_{u}}+1 }} 
{{\Theta }_{1}}\left( {{l}_{Y}},{{\chi }_{u}}+1,{{\eta }_{Y}};\frac{1}{\varepsilon {{\lambda }_{\textit{\text{R\textsuperscript{*}}D}}}},{{\tilde{\eta }}_{u}} \right), \\
	\end{aligned}
	\label{eq:J2}
\end{equation}
where 
${{\tilde{\gamma }}_{SE}}=\frac{{{\gamma }_{S}}{{\lambda }_{\textit{\text{SE}}}}}{{{\eta }_{u}}{{\gamma }_{S}}{{\lambda }_{\textit{\text{SE}}}}+{{\lambda }_{\textit{\text{SR\textsuperscript{*}}}}}}$ 
and 
${{\tilde{\eta }}_{u}}=\frac{{{\eta }_{u}}{{\tilde{\gamma }}_{SE}}}{\varepsilon {{\lambda }_{\textit{\text{R\textsuperscript{*}}D}}}}$.
The integral $\Theta_1$ is defined as
\begin{equation}
	\begin{aligned}
		{{\Theta }_{1}}\left( v,\gamma ,\mu ;\alpha ,\beta  \right)=\int_{0}^{\infty }{{{x}^{v}}{{e}^{-\mu x}}}{{\left( \frac{x+\alpha }{x+\beta } \right)}^{\gamma }}dx.
	\end{aligned}
	\label{eq:Theta1_def}
\end{equation}
By changing variable and using power series, we rewrite \eqref{eq:Theta1_def}
\begin{equation}
	\begin{aligned}
		{{\Theta }_{1}}&\left( v,\gamma ,\mu ;\alpha ,\beta  \right)= {{e}^{\mu \beta }}\sum\limits_{m=0}^{v}{\left( \begin{aligned}
				& v \\ 
				& m \\ 
		\end{aligned} \right)}
		{{\left( -\beta  \right)}^{v-m}} \sum\limits_{n=0}^{\gamma }
		{\left( \begin{aligned}
				& \gamma  \\ 
				& n \\ 
	    \end{aligned} \right)}
		{{\left( \alpha -\beta  \right)}^{\gamma -n}}  
		\int_{\beta }^{\infty }{{{e}^{-\mu t}}}{{t}^{m+n-\gamma }}dt. \\ 		
	\end{aligned}
	\label{eq:Theta1_1}
\end{equation}
We denote 
\begin{equation}
	{{\Phi }_{1}}\left( u;v,\mu  \right)=\int_{u}^{\infty }{{{x}^{v}}{{e}^{-\mu x}}}dx.
\end{equation}
Employing \cite[eq. (3.351,2) and eq. (3.351,4)]{zwillinger_3-4_2015} to express $\Phi_1$ as in \eqref{eq:Phi} for general case of arbitrary $v$, we have the expression of $\Theta_1$ as in \eqref{eq:Theta_1}.

The second integral can be expressed as
\begin{equation}
	\begin{aligned}
		{{J}_{2}}
		& ={{\left( \frac{\tilde{\gamma_Z}}{{\gamma }_{S}} \right)}^{{{q}_{Z}}}} 
		{{\Theta }_{2}}\left( {{l}_{Y}},{{q}_{Z}},{{\eta }_{Y}},{{\eta }_{Z}}\frac{\tilde{\gamma_Z}}{{\gamma }_{S}};{\tilde{\gamma_Y}} \right) ,\\
	\end{aligned}
	\label{eq:J3}
\end{equation}
where ${\tilde{\gamma_Y}}=\frac{{{\gamma }_{S}}}{\varepsilon {{\lambda }_{\textit{\text{R\textsuperscript{*}}D}}}\left( {{\gamma }_{S}}-1 \right)}$ 
and 
${\tilde{\gamma_Z}}=\frac{{{\gamma }_{S}}}{\varepsilon {{\lambda }_{\textit{\text{R\textsuperscript{*}}E}}}\left( {{\gamma }_{S}}-1 \right)}$. 
The integral $\Theta_2$ is defined as
\begin{equation}
	\begin{aligned}
		{{\Theta }_{2}}\left( v,\gamma ,\mu ,\rho ;\beta  \right)=\int_{0}^{\infty }{{{x}^{v}}{{\left( \frac{x}{x+\beta } \right)}^{\gamma }}{{e}^{-\mu x}}{{e}^{-\rho \left( \frac{x}{x+\beta } \right)}}}dx.
	\end{aligned}
	\label{eq:Theta2_def}
\end{equation}
By changing variable, using power series and Taylor expansion, we rewrite
\begin{equation}
	\begin{aligned}
		{{\Theta }_{2}}\left( v,\gamma ,\mu ,\rho ;\beta  \right)&={{e}^{\mu \beta -\rho }}\sum\limits_{m=0}^{v+\gamma }{\left( \begin{aligned}
				& v \\ 
				& m \\ 
			\end{aligned} \right)}{{\left( -\beta  \right)}^{v+\gamma -m}}\sum\limits_{p=0}^{\infty }{\frac{{{\left( \rho \beta  \right)}^{p}}}{p!}}
			\int_{\beta }^{\infty }{{{t}^{m-\gamma -p}}{{e}^{-\mu t}}}dt	.	
	\end{aligned}
	\label{eq:Theta2_1}
\end{equation}
Using expression of $\Phi_1$ as in \eqref{eq:Phi}, we have the expression of $\Theta_2$ as in \eqref{eq:Theta_2}.

Employing \cite[eq. (3.471,9)]{zwillinger_3-4_2015}, the third integral is written as
\begin{equation}
	\begin{aligned}
{{J}_{3}}
	& ={{\left( {{\tilde{\gamma }}_{SE}} \right)}^{{{\chi }_{u}}+1 }}
	{{\left(\frac {\tilde{\gamma_Z}}{{\gamma }_{S}} \right)}^{{{q}_{Z}}}} 
	{{\Theta }_{3}}\left( {{l}_{Y}},{{\chi }_{u}}+1,{{q}_{Z}},{{\eta }_{Y}},{{\eta }_{Z}}\frac{\tilde{\gamma_Z}}{{\gamma }_{S}};\frac{1}{\varepsilon {{\lambda }_{\textit{\text{R\textsuperscript{*}}D}}}},{{\tilde{\eta }}_{u}} \right) .\\ 
	\end{aligned}
	\label{eq:J4}
\end{equation}
The integral $\Theta_3$ is defined as
\begin{equation}
	\begin{aligned}
		&{{\Theta }_{3}}\left( v,\gamma ,\lambda ,\mu ,\rho ;\alpha ,\beta ,\xi  \right) ={{\int_{0}^{\infty }{{{x}^{v}}{{\left( \frac{x+\alpha }{x+\beta } \right)}^{\gamma }}\left( \frac{x}{x+\xi } \right)}}^{\lambda }}{{e}^{-\rho \left( \frac{x}{x+\beta } \right)}}{{e}^{-\mu x}}dx ,\\
	\end{aligned}
	\label{eq:Theta3_def}
\end{equation}
Applying Taylor expansion and partial fraction in \cite[eq. (2.102)]{zwillinger_2_2014}, $\Theta_3$ can be expressed as
\begin{equation}
	\begin{aligned}
		{{\Theta }_{3}}&\left( v,\gamma ,\lambda ,\mu ,\rho ;\alpha ,\beta ,\xi  \right)=\sum\limits_{p=0}^{\infty }{\frac{{{\left( -\rho  \right)}^{p}}}{p!}}\sum\limits_{m=0}^{\gamma }{\left( \begin{aligned}
				& \gamma  \\ 
				& m \\ 
			\end{aligned} \right)}{{\alpha }^{\gamma -m}} \\ 
		& \quad\quad \times \left\{ \begin{aligned}
			& \sum\limits_{a=1}^{\gamma +p}{{{A}_{\gamma +p-a+1}}\int_{0}^{\infty }{\frac{{{e}^{-\mu x}}}{{{\left( x+\beta  \right)}^{\gamma +p-a+1}}}}dx}  +\sum\limits_{b=1}^{\lambda }{{{B}_{\lambda -b+1}}\int_{0}^{\infty }{\frac{{{e}^{-\mu x}}}{{{\left( x+\xi  \right)}^{\lambda -b+1}}}}dx} \\ 
		\end{aligned} \right\} ,\\ 
	\end{aligned}
	\label{eq:Theta3_1}
\end{equation}
where ${{A}_{\gamma +p-a+1}}$ and ${B}_{\lambda -b+1}$ as in \eqref{eq:Theta_3}, and 
\begin{equation}
	{{\Phi }_{2}}\left( \gamma ,\mu ;\beta  \right)=\int_{0}^{\infty }{\frac{{{e}^{-\mu x}}}{{{\left( x+\beta  \right)}^{\gamma }}}}dx.
	\label{eq:Phi2}
\end{equation}
Employing \cite[eq. (3.351,2) and eq. (3.353,2)]{zwillinger_3-4_2015} to express $\Phi_2$ as in \eqref{eq:Phi}, we derive the expression of $\Theta_3$ as in \eqref{eq:Theta_3}.
Inserting all integrals back to \eqref{eq:SOP_SC3}, we obtain \eqref{eq:SOP_SC_final} to complete the proof.
\begin{figure*}[!tb]
	\begin{center}
		\begin{equation}
			\begin{aligned}
				P_{\textit{\text{out}}}^{{\text{SC,J}}}\left( {{\gamma }_{S}} \right)
				&=1-U\sum\limits_{{{l}_{X}}=0}^{{{m}_{S}}-1}{{}}\widetilde{\sum }_{u}^{U-1}\widehat{\sum }_{k}^{U-1}\sum\limits_{{{l}_{Z}}=0}^{{{m}_{S}}-1}{{{\zeta }_{X}}{{\zeta }_{Z}}}{{\kappa }_{u}}\Gamma \left( {{\chi }_{u}}+1 \right){{\zeta }_{k}} 
				\int_{0}^{\infty }{{{t}^{\chi _{k}^{M-1}-1}}{{e}^{-{{\eta }_{Z}}t}}} \\ 
				& \times  \left\{ \begin{aligned}
					& {{\left( {{\eta }_{u}} \right)}^{-\left( {{\chi }_{u}}+1 \right)}}
					\underbrace{\sum\limits_{{{l}_{Y}}=0}^{{{m}_{S}}-1}{{{\zeta }_{Y}}}\int_{0}^{{{\Upsilon }_{Z}}\left( t \right)}{{{z}^{{{l}_{Z}}}}{{e}^{-{{\eta }_{Z}}z}}}\int_{{{\Upsilon }_{Y}}\left( z,t \right)}^{\infty }{{y}^{{{l}_{Y}}}}{{e}^{-{{\eta }_{Y}}y}}dydz}_{{{J}_{4}}} dt \\ 
					& -\sum\limits_{{{l}_{Y}}=0}^{{{m}_{S}}-1}{{{\zeta }_{Y}}}
					\int_{0}^{{{\Upsilon }_{Z}}\left( t \right)}{{{z}^{{{l}_{Z}}}}{{e}^{-{{\eta }_{Z}}z}}}
					\underbrace{\int_{{{\Upsilon }_{Y}}\left( z,t \right)}^{\infty }{{{y}^{{{l}_{Y}}}}{{e}^{-{{\eta }_{y}}y}}}{{\left( {{\eta }_{u}}+{{\Upsilon }_{W}}\left( y \right) \right)}^{-\left( {{\chi }_{u}}+1 \right)}}dy}_{{{J}_{5}}} dzdt\\ 
				\end{aligned} \right\} .\\ 				
			\end{aligned}		
			\label{eq:SOP_SCJ3}	
		\end{equation}
	\end{center}
	\setlength{\arraycolsep}{1pt}
	\hrulefill \setlength{\arraycolsep}{0.0em}
\end{figure*}
%
\section{Proof of Proposition 2}
\label{app:Pro2}
Equation \eqref{eq:SOP_SCJ1} is rewritten using the CDF and PDF of all related r.v. to have \eqref{eq:SOP_SCJ3} after some simplifications. 
Employing \cite[eq. (3.471,9)]{zwillinger_3-4_2015}, we write the fourth integral as
\begin{equation}
	\begin{aligned}
		{{J}_{4}}
		=\sum\limits_{{{l}_{Y}}=0}^{{{m}_{S}}-1}{\sum\limits_{{{q}_{Y}}=0}^{{{l}_{Y}}}}{{{\kappa }_{Y}}}&{{\left( {{\tilde{\gamma }}_{Y}} \right)}^{{{q}_{Y}}}} 
		{{\Theta }_{4}}\left( {{\Upsilon }_{Z}}\left( t \right);{{l}_{Z}},{{q}_{Y}},{{\eta }_{Z}},{{\eta }_{Y}}{{{\tilde{\gamma }}_{Y}}};{{\Upsilon }_{Z}}\left( t \right) \right) .\\
	\end{aligned}
	\label{eq:J5}
\end{equation}
The integral $\Theta_4$ is defined as
\begin{equation}
	\begin{aligned}
		{{\Theta }_{4}}\left( u;v,\gamma ,\mu ,\rho ;\beta  \right)=\int_{0}^{u}{{{x}^{v}}{{\left( \frac{x}{\beta -x} \right)}^{\gamma }}{{e}^{-\mu x}}{{e}^{-\rho \left( \frac{x}{\beta -x} \right)}}}dx .
	\end{aligned}
	\label{eq:Theta4_def}
\end{equation}
By changing variable and using power series, we rewrite
\begin{equation}
	\begin{aligned}
		{{\Theta }_{4}}\left( u;v,\gamma ,\mu ,\rho ;\beta  \right)
		={{e}^{\rho -\mu \beta }}
		& \sum\limits_{m=0}^{v+\gamma }{\left( \begin{aligned}
				& v+\gamma  \\ 
				& m \\ 
			\end{aligned} \right)}{{\left( -1 \right)}^{m}}{{\beta }^{v+\gamma -m}} \\ 
		& \times \sum\limits_{n=0}^{\infty }{\frac{{{\left( \mu  \right)}^{n}}}{n!}\sum\limits_{p=0}^{\infty }{\frac{{{\left( -1 \right)}^{p}}{{\left( \rho \beta  \right)}^{p}}}{p!}}}\int_{\beta -u}^{\beta }{{{t}^{\gamma +m+n-p}}}dt .\\ 
	\end{aligned}
	\label{eq:Theta4_1}
\end{equation}
The expression of $\Theta_4$ is then derived as in \eqref{eq:Theta_4}.

Employing \cite[eq. (3.471,9)]{zwillinger_3-4_2015}, the fifth integral
\begin{equation}
	\begin{aligned}
		{{J}_{5}}
		&={{\left( {{{\tilde{\gamma }}}_{S}} \right)}^{{{\chi }_{u}}+1}} {{\Theta }_{5}}\left( {{\Upsilon }_{Y}}\left( z,t \right);{{l}_{Y}},{{\chi }_{u}}+1,{{\eta }_{Y}};\frac{1}{\varepsilon {{\lambda }_{\textit{\text{R\textsuperscript{*}}D}}}},{{{\tilde{\eta }}}_{u}} \right) .\\ 
	\end{aligned}
	\label{eq:J6}
\end{equation}
The integral $\Theta_5$ is defined as
\begin{equation}
	\begin{aligned}
		{{\Theta }_{5}}\left( u;v,\gamma ,\mu ;\alpha ,\beta  \right)=\int_{u}^{\infty }{{{x}^{v}}}{{\left( \frac{x+\alpha }{x+\beta } \right)}^{\gamma }}{{e}^{-\mu x}}dx .
	\end{aligned}
	\label{eq:Theta5_def}
\end{equation}
By changing variable and using power series, we rewrite
\begin{equation}
	\begin{aligned}
		{{\Theta }_{5}}\left( u;v,\gamma ,\mu ;\alpha ,\beta  \right)
		& ={{e}^{\mu \beta }}\sum\limits_{m=0}^{v}
		{\left(\begin{aligned}
				& v \\ 
				& m \\ 
			\end{aligned} \right)}
		{{\left( -\beta  \right)}^{v-m}} \sum\limits_{n=0}^{\gamma }
		{\left( \begin{aligned}
				& \gamma  \\ 
				& n \\ 
			\end{aligned} \right)}
		{{\left( \alpha -\beta  \right)}^{\gamma -n}}\int_{u+\beta }^{\infty }{{{e}^{-\mu t}}}{{t}^{m+n-\gamma }}dt .\\ 		
	\end{aligned}
	\label{eq:Theta5_1}
\end{equation}
Using expression of $\Phi_1$ as in \eqref{eq:Phi}, the expression of $\Theta_5$ is then derived as in \eqref{eq:Theta_5}.
Inserting all integrals back to \eqref{eq:SOP_SCJ3}, we obtain \eqref{eq:SOP_SCJ_final} to complete the proof.
\begin{figure*}[!tb]
	\begin{center}
		\begin{equation}
			\begin{aligned}
%
	P_{\textit{\text{out}}}^{{\text{MRC}}}\left( {{\gamma }_{S}} \right)
	= 									& \sum\limits_{{{l}_{Y}}=0}^{{{m}_{S}}-1}
	\sum\limits_{{{l}_{Z}}=0}^{{{m}_{S}}-1} {{\zeta }_{Y}}
	{\sum\limits_{{{q}_{Z}}=0}^{{{l}_{Z}}}{\kappa }_{Z}}
		\underbrace{\int_{0}^{\infty }{{{y}^{{{l}_{Y}}}}{{e}^{-{{\eta }_{Y}}y}}}{{\left\{ {{\Upsilon }_{Y}}\left( y \right) \right\}}^{{{q}_{Z}}}}{{e}^{-{{\eta }_{Z}}{{\Upsilon }_{Y}}\left( y \right)}}dy}_{{{J}_{6}}} \\ 
	&  \begin{aligned}
	    + U \sum\limits_{{{l}_{X}}=0}^{{{m}_{S}}-1}
	    &{\widetilde{\sum }_{u}^{U-1}}
	    \sum\limits_{{{l}_{Y}}=0}^{{{m}_{S}}-1}{{}}
	    \sum\limits_{{{l}_{Z}}=0}^{{{m}_{S}}-1}
	    {{{\zeta }_{X}}{{\zeta }_{Y}}{{\zeta }_{Z}}}{{\kappa }_{u}}\Gamma \left( {{\chi }_{u}}+1 \right) \\
	    & \times \int_{0}^{\infty }{{{y}^{{{l}_{Y}}}}{{e}^{-{{\eta }_{Y}}y}}}
		\underbrace{\int_{0}^{{{\Upsilon }_{Y}}\left( y \right)}{{{\left( {{\eta }_{u}}+{{\Upsilon }_{W}}\left( y,z \right) \right)}^{-\left( {{\chi }_{u}}+1 \right)}}{{z}^{{{l}_{Z}}}}{{e}^{-{{\eta }_{Z}}z}}dz}}_{{{J}_{7}}} dy .
	\end{aligned} \\ 
			\end{aligned} 
			\label{eq:SOP_MRC3}
		\end{equation}	
	\end{center}
	\setlength{\arraycolsep}{1pt}
	\hrulefill \setlength{\arraycolsep}{0.0em}
\end{figure*}
\section{Proof of Proposition 3}
\label{app:Pro3}
Rewriting \eqref{eq:SOP_MRC1} by using the CDF and PDF of all r.v. and after some simplifications, we have \eqref{eq:SOP_MRC3}. 
Employing \cite[eq. (3.471,9)]{zwillinger_3-4_2015}, we have
\begin{equation}
	\begin{aligned}
		{{J}_{6}}
		&={{\left( \frac{\tilde{\gamma_Z}}{\gamma_S} \right)}^{{q}_{Z}}}
	    {{\Theta }_{2}}\left( {{l}_{Y}},{{q}_{Z}},{{\eta }_{Y}},{{\eta }_{Z}}{\tilde{\gamma_Z}};\tilde{\gamma_Y} \right) ,\\
	\end{aligned}
	\label{eq:J7}
\end{equation}
%
\begin{equation}
	\begin{aligned}
		{{J}_{7}}
		&={{\left( \frac{{{\lambda }_{\textit{\text{SE}}}}}{{{{\tilde{\lambda }}}_{Y}}{{\lambda }_{\textit{\text{SE}}}}-{{\lambda }_{\textit{\text{SR\textsuperscript{*}}}}}} \right)}^{{{\chi }_{u}}+1}} 
		{{\Theta }_{6}}\left( {{\Upsilon }_{Y}}\left( y \right);{{l}_{Z}},{{\chi }_{u}}+1,{{\eta }_{Z}};\frac{1}{\varepsilon {{\lambda }_{\textit{\text{R\textsuperscript{*}}E}}}},\frac{{\tilde{\lambda }}_{Y}}{\varepsilon {{\lambda }_{\textit{\text{R\textsuperscript{*}}E}}}} \right) ,\\ 
	\end{aligned}
	\label{eq:J8}
\end{equation}
where ${{\tilde{\lambda }}_{Y}}={{\eta }_{u}}+\frac{\varepsilon {{\lambda }_{\textit{\text{SR\textsuperscript{*}}}}}{ \lambda_{\textit{\text{R\textsuperscript{*}}D}}}y}{{{\lambda }_{\textit{\text{SE}}}}{{\gamma }_{S}}\left( \varepsilon {{\lambda }_{\textit{\text{R\textsuperscript{*}}D}}}y+1 \right)}$.
The integral $\Theta_6$ is defined as
\begin{equation}
	{{\Theta }_{6}}\left( u;v,\gamma ,\mu ;\alpha ,\beta  \right)=\int_{0}^{u}{{{x}^{v}}{{\left( \frac{x+\alpha }{x+\beta } \right)}^{\gamma }}}{{e}^{-\mu x}}dx .
	\label{eq:Theta6_def}
\end{equation}
By changing variable and using power series, we rewrite
\begin{equation}
	\begin{aligned}
		{{\Theta }_{6}}\left( u;v,\gamma ,\mu ;\alpha ,\beta  \right)={{e}^{\mu \beta }}\sum\limits_{m=0}^{v}{\left( \begin{aligned}
				& v \\ 
				& m \\ 
			\end{aligned} \right)}{{\left( -\beta  \right)}^{v-m}} \sum\limits_{n=0}^{\gamma }{\left( \begin{aligned}
				& \gamma  \\ 
				& n \\ 
			\end{aligned} \right)}{{\left( \alpha -\beta  \right)}^{\gamma -n}}\int_{\beta }^{u+\beta }{{{e}^{-\mu t}}}{{t}^{m+n-\gamma }}dt .\\ 
	\end{aligned}
	\label{eq:Theta6_1}
\end{equation}
We derive the expression of $\Theta_6$ as in \eqref{eq:Theta_6} using expression of $\Phi_1$ as in \eqref{eq:Phi}.
Inserting all integrals back to \eqref{eq:SOP_MRC3}, we obtain \eqref{eq:SOP_MRC_final} to complete the proof.
\begin{figure*}[!tb]
	\begin{center}
		\begin{equation}
			\begin{aligned}
	& P_{\textit{\text{out}}}^{{\text{MRC,J}}}\left( {{\gamma }_{S}} \right)=\sum\limits_{{{l}_{Z}}=0}^{{{m}_{S}}-1}{\sum\limits_{{{q}_{Z}}=0}^{{{l}_{Z}}}{\widehat{\sum }_{k}^{U-1}{{\kappa }_{Z}}}}{{\zeta }_{k}}
	\underbrace{\int_{0}^{\infty }{{{\left\{ {{\Upsilon }_{Z}}\left( t \right) \right\}}^{{{q}_{Z}}}}{{e}^{-{{\eta }_{Z}}{{\Upsilon }_{Z}}\left( t \right)}}{{t}^{\chi _{k}^{U-1}-1}}{{e}^{-{{\eta }_{Z}}t}}}dt}_{{{J}_{8}}} \\ 
	& {\begin{aligned}
	+\sum\limits_{{{l}_{Z}}=0}^{{{m}_{S}}-1}{{}} & \widehat{\sum }_{k}^{U-1}{{\zeta }_{Z}}{{\zeta }_{k}} 
	\Biggl\{             
	    \underbrace{\int_{0}^{\infty }{{{t}^{\chi _{k}^{U-1}-1}}{{e}^{-{{\eta }_{Z}}t}}}\int_{0}^{{{\Upsilon }_{Z}}\left( t \right)}{{{z}^{{{l}_{Z}}}}{{e}^{-{{\eta }_{Z}}z}}}dzdt}_{{{J}_{9}}} \\ 
	    & -\sum\limits_{{{l}_{Y}}=0}^{{{m}_{S}}-1}{\sum\limits_{{{q}_{Y}}=0}^{{{l}_{Y}}}{{{\kappa }_{Y}}}}\int_{0}^{\infty }{{{t}^{\chi _{k}^{U-1}-1}}{{e}^{-{{\eta }_{Z}}t}}}
		\underbrace{\int_{0}^{{{\Upsilon }_{Z}}\left( t \right)}{{{z}^{{{l}_{Z}}}}{{e}^{-{{\eta }_{Z}}z}}}{{\left\{ {{\Upsilon }_{Y}}\left( z,t \right) \right\}}^{{{q}_{Y}}}}{{e}^{-{{\eta }_{Y}}{{\Upsilon }_{Y}}\left( z,t \right)}}dz}_{{{J}_{10}}}dt \Biggr\} \\
	\end{aligned}} \\ 
	& +U\sum\limits_{{{l}_{X}}=0}^{{{m}_{S}}-1}{{}}\widetilde{\sum }_{u}^{U-1}\sum\limits_{{{l}_{Y}}=0}^{{{m}_{S}}-1}{{}}\sum\limits_{{{l}_{Z}}=0}^{{{m}_{S}}-1}{\widehat{\sum }_{k}^{U-1}{{\zeta }_{X}}{{\zeta }_{Y}}{{\zeta }_{Z}}{{\kappa }_{u}}}{{\zeta }_{k}}\Gamma \left( {{\chi }_{u}}+1 \right) \\ 
		& \quad \times \int_{0}^{\infty }{{{t}^{\chi _{k}^{U-1}-1}}{{e}^{-{{\eta }_{Z}}t}}}\int_{0}^{{{\Upsilon }_{Z}}\left( t \right)}{{{z}^{{{l}_{Z}}}}{{e}^{-{{\eta }_{Z}}z}}}
		\underbrace{\int_{{{\Upsilon }_{Y}}\left( z,t \right)}^{\infty }{{{y}^{{{l}_{Y}}}}{{e}^{-{{\eta }_{Y}}y}}}{{\left\{ {{\eta }_{u}}+{{\Upsilon }_{W}}\left( y,z,t \right) \right\}}^{-\left( {{\chi }_{u}}+1 \right)}}dy}_{{{J}_{11}}}dzdt .\\				
			\end{aligned} 
			\label{eq:SOP_MRCJ3}
		\end{equation}		
	\end{center}
	\setlength{\arraycolsep}{1pt}
	\hrulefill \setlength{\arraycolsep}{0.0em}
\end{figure*}
\section{Proof of Proposition 4}
\label{app:Pro4}
The CDF and PDF of all r.v. with friendly jamming are put into \eqref{eq:SOP_MRCJ2} to have the SOP with respect to the MRC scheme at $E$ as in \eqref{eq:SOP_MRCJ3}. Each integral is continued to be formulated as follows:

\begin{equation}
	\begin{aligned}
{{J}_{8}}
&={{\left( \frac{{{P}_{J}}}{\varepsilon \left( {{\gamma }_{S}}-1 \right)} \right)}^{{{q}_{Z}}}}
{{\Theta }_{7}}\left( \chi _{k}^{U-1}-1,{{q}_{Z}},{{\eta }_{Z}},\frac{{{\eta }_{Z}}{{P}_{J}}}{\varepsilon \left( {{\gamma }_{S}}-1 \right)};\frac{1}{{{P}_{J}}{{\lambda }_{{{R}^{*}}E}}} \right) ,\\
	\end{aligned}
	\label{eq:J9}
\end{equation}

\begin{equation}
	\begin{aligned}
{{J}_{9}}
&=\Gamma \left( {{l}_{Z}}+1 \right)\Gamma \left( \chi _{k}^{U-1} \right){{\eta }_{Z}}^{-\left( {{l}_{Z}}+1+\chi _{k}^{U-1} \right)} \\ 
& -\sum\limits_{{q_Z}=0}^{{{l}_{Z}}}{\frac{{{l}_{Z}}!}{{q_Z}!}}\frac{1}{{{\eta }_{Z}}^{{{l}_{Z}}-{q_Z}+1}}{{\left( \frac{{{P}_{J}}}{\varepsilon \left( {{\gamma }_{S}}-1 \right)} \right)}^{{q_Z}}} 
    {{\Theta }_{7}}\left( \chi _{k}^{U-1}-1,{q_Z},{{\eta }_{Z}},\frac{{{\eta }_{Z}}{{P}_{J}}}{\varepsilon \left( {{\gamma }_{S}}-1 \right)};\frac{1}{{{P}_{J}}{{\lambda }_{\textit{\text{R\textsuperscript{*}}E}}}} \right) .\\
	\end{aligned}
	\label{eq:J10}
\end{equation}
The integral $\Theta_7$ is defined as
\begin{equation}
	{{\Theta }_{7}}\left( v,\gamma ,\mu ,\rho ;\alpha  \right)=\int_{0}^{\infty }{{{x}^{v}}{{\left( x+\alpha  \right)}^{\gamma }}}{{e}^{-\mu x}}{{e}^{-\rho \left( x+\alpha  \right)}}dx .
	\label{eq:Theta7_def}
\end{equation}
And we have
\begin{equation}
	\begin{aligned}
{{J}_{10}}
    & = {{\left( \frac{{{\gamma }_{S}}}{\varepsilon {{\lambda }_{\textit{\text{R\textsuperscript{*}}D}}}\left( {{\gamma }_{S}}-1 \right)} \right)}^{{{q}_{Y}}}}  {{\Theta }_{4}}\left( {{\Upsilon }_{Z}}\left( t \right);{{l}_{Z}},{{q}_{Y}},{{\eta }_{Z}},{{{\tilde{\eta }}}_{Y}};{{\Upsilon }_{Z}}\left( t \right) \right) .\\
	\end{aligned}
	\label{eq:J11}
\end{equation}

\begin{equation}
	\begin{aligned}
{{J}_{11}}
& ={{\left( \frac{{{\gamma }_{S}}}{{{\lambda }_{\textit{\text{SR\textsuperscript{*}}}}}+{{\gamma }_{S}}{{\lambda }_{\textit{\text{SE}}}}\tilde{V}\left( z,t \right)} \right)}^{{{\chi }_{u}}+1}} 
    {{\Theta }_{5}}\left( {{\Upsilon }_{Y}}\left( z,t \right);{{l}_{Y}},{{\chi }_{u}}+1,{{\eta }_{Y}};\frac{1}{\varepsilon {{\lambda }_{\textit{\text{R\textsuperscript{*}}D}}}},\beta_5 \right) , \\  
	\end{aligned}
	\label{eq:J12}
\end{equation}
where $\beta_5=\frac{{{\gamma }_{S}}{{\lambda }_{\textit{\text{SE}}}}\tilde{V}\left( z,t \right)}{\varepsilon {{\lambda }_{\textit{\text{R\textsuperscript{*}}D}}}\left( {{\lambda }_{\textit{\text{SR\textsuperscript{*}}}}}+{{\gamma }_{S}}{{\lambda }_{\textit{\text{SE}}}}\tilde{V}\left( z,t \right) \right)}$, and $\tilde{V}\left( z,t \right)={{\eta }_{u}}-\frac{\varepsilon {{\lambda }_{\textit{\text{SR\textsuperscript{*}}}}}}{{{\lambda }_{\textit{\text{SE}}}}}V\left( z,t \right)$.
Inserting all integrals back to \eqref{eq:SOP_MRCJ3}, we obtain \eqref{eq:SOP_MRCJ_final} to complete the proof.
\section{Lower-Bound Analysis}
\begin{figure*}[!tb]
	\begin{center}
		\begin{equation}
			\begin{aligned}
				&{\bar{P}}_{\textit{\text{out}}}^{{\text{SC}}}\left( {{\gamma }_{S}} \right)=1-U\sum\limits_{{{l}_{X}}=0}^{{{m}_{S}}-1}{\widetilde{\sum }_{u}^{U-1}}\sum\limits_{{{l}_{Y}}=0}^{{{m}_{S}}-1}{{{\zeta }_{X}}{{\zeta }_{Y}}}{{\kappa }_{u}}\Gamma \left( {{\chi }_{u}}+1 \right) \\ 
				& \times \Biggl\{ {{\eta }_{u}}^{-\left( {{\chi }_{u}}+1 \right)}\Gamma \left( {{l}_{Y}}+1 \right){{\eta }_{Y}}^{-\left( {{l}_{Y}}+1 \right)}
				-{\underbrace{\int_{0}^{\infty }{{{y}^{{{l}_{Y}}}}{{e}^{-{{\eta }_{Y}}y}}}{{\mathbb{E}}_{\mathbf{p_E}}}\left\{{\left( {{\eta }_{u}}+{{\Upsilon }_{W}}\left( y,{d_{\textit{\text{S}E}}} \right) \right)}^{-\left( {{\chi }_{u}}+1 \right)}\right\}dy}_{{\bar{J}_{1}}}} \\ 
				& -\sum\limits_{{{l}_{Z}}=0}^{{{m}_{S}}-1}{\sum\limits_{{{q}_{Z}}=0}^{{{l}_{Z}}}{{{\kappa }_{Z}}}}{{\eta }_{u}}^{-\left( {{\chi }_{u}}+1 \right)}
				{\underbrace{\int_{0}^{\infty }{{{y}^{{{l}_{Y}}}}{{e}^{-{{\eta }_{Y}}y}}{{{{\mathbb{E}}_{\mathbf{p_E}}}\left\{ {{\Upsilon }_{Z}}\left( y,{d_{\textit{\text{R\textsuperscript{*}}E}}} \right)^{{q}_{Z}}
				{{e}^{-{{\eta }_{Z}}{{\Upsilon }_{Z}}\left( y,{d_{\textit{\text{R\textsuperscript{*}}E}}} \right)}} \right\}}}}dy}_{{\bar{J}_{2}}}} \\ 
				& -\sum\limits_{{{l}_{Z}}=0}^{{{m}_{S}}-1}{\sum\limits_{{{q}_{Z}}=0}^{{{l}_{Z}}}{{{\kappa }_{Z}}}} {\underbrace{\int_{0}^{\infty }{{{y}^{{{l}_{Y}}}}{{e}^{-{{\eta }_{Y}}y}}} {{\mathbb{E}}_{\mathbf{p_E}}}\left\{ {{\left( {{\eta }_{u}}+{{\Upsilon }_{W}}\left( y,{d_{\textit{\text{S}E}}} \right) \right)}^{-\left( {{\chi }_{u}}+1 \right)}}{{\Upsilon }_{Z}}\left( y,{d_{\textit{\text{R\textsuperscript{*}}E}}} \right) ^{{q}_{Z}}
				{{e}^{-{{\eta }_{Z}}{{\Upsilon }_{Z}}\left( y,{d_{\textit{\text{R\textsuperscript{*}}E}}} \right)}} \right\} dy}_{{\bar{J}_{3}}}} 
				\Biggr\} .\\ 
			\end{aligned} 
			\label{eq:SOP_SC_E1}
		\end{equation}
	\end{center}
	\hrulefill \setlength{\arraycolsep}{0.0em}
\end{figure*}
\begin{figure*}[!tb]
	\begin{center}
		\begin{equation}
			\begin{aligned}			
			& {\bar{P}}_{\textit{\text{out}}}^{{\text{SC,J}}}
			\left( {{\gamma }_{S}} \right)=1-U \sum\limits_{{{l}_{X}}=0}^{{{m}_{S}}-1}{{}}\widetilde{\sum }_{u}^{U-1}         \widehat{\sum }_{k}^{U-1} \sum\limits_{{{l}_{Z}}=0}^{{{m}_{S}}-1} {{{\zeta }_{X}} {{\zeta }_{Z}}}{{\kappa }_{u}} \Gamma \left( {{\chi }_{u}}+1 \right){{\zeta }_{k}} 
			\int_{0}^{\infty }{{{t}^{\chi _{k}^{M-1}-1}}{{e}^{-{{\eta }_{Z}}t}}} \\ 
			& \times \Biggl\{ 
			{{\left( {{\eta }_{u}} \right)}^{-\left( {{\chi }_{u}}+1 \right)}} 
				\underbrace{\sum\limits_{{{l}_{Y}}=0}^{{{m}_{S}}-1}{{{\zeta }_{Y}}}
				{{\mathbb{E}}_{\mathbf{p_E}}} \left\{\int_{0}^{{{\Upsilon }_{Z}}\left( t,{d_{\textit{\text{R\textsuperscript{*}}E}}} \right)} {{{z}^{{{l}_{Z}}}}{{e}^{-{{\eta }_{Z}}z}}} \int_{{{\Upsilon }_{Y}}\left( z,t,{d_{\textit{\text{R\textsuperscript{*}}E}}} \right)}^{\infty } {{y}^{{{l}_{Y}}}} {{e}^{-{{\eta }_{Y}}y}}dydz\right\}}_{{\bar{J}_{4}}} \\ 
				& -\sum\limits_{{{l}_{Y}}=0}^{{{m}_{S}}-1}{{{\zeta }_{Y}}}  {{\mathbb{E}}_{\mathbf{p_E}}}
					\left\{\int_{0}^{{{\Upsilon }_{Z}}\left( t,{d_{\textit{\text{R\textsuperscript{*}}E}}} \right)}{{{z}^{{{l}_{Z}}}}{{e}^{-{{\eta }_{Z}}z}}} 
					\underbrace{\int_{{{\Upsilon }_{Y}}\left( z,t,{d_{\textit{\text{R\textsuperscript{*}}E}}} \right)}^{\infty }{{{y}^{{{l}_{Y}}}}{{e}^{-{{\eta }_{y}}y}}}{{\left( {{\eta }_{u}}+{{\Upsilon }_{W}}\left( y,{d_{\textit{\text{S}E}}} \right) \right)}^{-\left( {{\chi }_{u}}+1 \right)}}dy}_{{\bar{J}_{5}}} dz\right\} \Biggr\}dt . \\ 
			\end{aligned}		
			\label{eq:SOP_SCJ_E1}	
		\end{equation}
	\end{center}
	\setlength{\arraycolsep}{1pt}
	\hrulefill \setlength{\arraycolsep}{0.0em}
\end{figure*}
In our model, the SOP depends on $E$'s location via the distance from $S$ and $R$ to $E$. Therefore, by swapping the integration, we have \eqref{eq:SOP_SC_E1} and \eqref{eq:SOP_SCJ_E1}. Applying Jensen's inequality, one can obtain the lower bound expression, $\mathbb{E}\left[ g\left( X \right) \right]\ge g\left( \mathbb{E}\left[ X \right] \right)$, 
\begin{equation}
	\begin{aligned}
		{{\mathbb{E}}_{\mathbf{p_E}}} 
		& \left\{{\left( {{\eta }_{u}}+{{\Upsilon }_{W}}\left( y,{d_{\textit{\text{S}E}}} \right) \right)}^{-\left( {{\chi }_{u}}+1 \right)}\right\} 
		\ge 
		{\left( {{\eta }_{u}}+{{\Upsilon }_{W}}\left( y,{{\mathbb{E}}_{\mathbf{p_E}}}\left\{{d_{\textit{\text{S}E}}} \right\} \right) \right)}^{-\left( {{\chi }_{u}}+1 \right)} ,
	\end{aligned}
\end{equation}
%
\begin{equation}
	\begin{aligned}
	{{\mathbb{E}}_{\mathbf{p_E}}}
		& \left\{\int_{0}^{{{\Upsilon }_{Z}}\left( t,{d_{\textit{\text{R\textsuperscript{*}}E}}} \right)}{{{z}^{{{l}_{Z}}}}{{e}^{-{{\eta }_{Z}}z}}}\int_{{{\Upsilon }_{Y}}\left( z,t,{d_{\textit{\text{R\textsuperscript{*}}E}}} \right)}^{\infty }{{y}^{{{l}_{Y}}}}{{e}^{-{{\eta }_{Y}}y}}dydz 
		\right\} \\
		& \ge
		\int_{0}^{ {{\Upsilon }_{Z}}\left( t,{{\mathbb{E}}_{\mathbf{p_E}}}\left\{{d_{\textit{\text{R\textsuperscript{*}}E}}} \right\} \right) }
		{{{z}^{{{l}_{Z}}}}{{e}^{-{{\eta }_{Z}}z}}}
		\int_{{\Upsilon }_{Y}\left( z,t,{{\mathbb{E}}_{\mathbf{p_E}}}\left\{{d_{\textit{\text{R\textsuperscript{*}}E}}} \right\} \right)}^{\infty }{{y}^{{{l}_{Y}}}}{{e}^{-{{\eta }_{Y}}y}}dydz .
	\end{aligned}
\end{equation}

Above process is applied to other integrals.
We denote ${\mathcal{R}_{\textit{\text{S}E}}}={{\mathbb{E}}_{\mathbf{p_E}}}\left\{ {d_{\textit{\text{S}E}}} \right\}$ and ${\mathcal{R}_{\textit{\text{R\textsuperscript{*}}E}}}={{\mathbb{E}}_{\mathbf{p_E}}}\left\{ {d_{\textit{\text{R\textsuperscript{*}}E}}} \right\}$ to obtain the expressions of the SOP in the general case of a randomly distributed eavesdropper.

\bibliographystyle{IEEEtran}
\bibliography{TCOM}

\end{document}